\RequirePackage[l2tabu,orthodox]{nag}
\documentclass
[11pt,letterpaper]
{article} 

\usepackage[notes=true,later=false,camera=false]{dtrt}
\usepackage[utf8]{inputenc}
\usepackage{etex}
\usepackage{ stmaryrd }
\usepackage{xspace,enumerate}
\usepackage[T1]{fontenc}
\usepackage[full]{textcomp}
\usepackage[american]{babel}
\usepackage{mathtools}
\usepackage{titling}

\usepackage{amsthm}
\usepackage{empheq}

   \usepackage{hyperref}
   \hypersetup{hyperindex=true,pdfpagemode=UseOutlines,bookmarksnumbered=true,bookmarksopen=true,bookmarksopenlevel=2,pdfstartview=FitH,pdfborder={0 0 1},linkbordercolor=dt@linkcolor,citebordercolor=dt@linkcolor,urlbordercolor=dt@linkcolor}
\usepackage[capitalise,nameinlink]{cleveref}
\crefname{lemma}{Lemma}{Lemmas}
\crefname{fact}{Fact}{Facts}
\newcommand{\colorconstraints}{\text{Color Constraints}}
\crefname{colorconstraints}{(color constraints)}{Color Constraints}
\crefformat{colorconstraints}{#2\colorconstraints#3}
\crefname{indsetconstraints}{(indset constraints)}{IndSet Constraints}
\crefformat{indsetconstraints}{#2$\mathsf{IndSet\ Axioms}$#3}
\crefname{theorem}{Theorem}{Theorems}
\crefname{mtheorem}{Theorem}{Theorems}
\crefname{corollary}{Corollary}{Corollaries}
\crefname{claim}{Claim}{Claims}
\crefname{example}{Example}{Examples}
\crefname{algorithm}{Algorithm}{Algorithms}
\crefname{problem}{Problem}{Problems}
\crefname{definition}{Definition}{Definitions}
\usepackage{paralist}
\usepackage{turnstile}
\usepackage{mdframed}
\usepackage{tikz}
\usepackage{caption}
\DeclareCaptionType{Algorithm}
\usepackage{newfloat}
\newtheorem{theorem}{Theorem}[section]
\newtheorem{mtheorem}{Theorem}
\newtheorem*{theorem*}{Theorem}

\newtheorem*{proposition*}{Proposition}
\newtheorem{lemma}[theorem]{Lemma}
\newtheorem*{lemma*}{Lemma}

\newtheorem*{conjecture*}{Conjecture}
\newtheorem{fact}[theorem]{Fact}
\newtheorem*{fact*}{Fact}

\newtheorem*{hypothesis*}{Hypothesis}

\theoremstyle{definition}
\newtheorem{definition}[theorem]{Definition}
\newtheorem*{definition*}{Definition}

\theoremstyle{remark}
\newtheorem{claim}[theorem]{Claim}
\newtheorem*{claim*}{Claim}
\newtheorem{remark}[theorem]{Remark}
\newtheorem*{remark*}{Remark}

\newtheorem*{observation*}{Observation}

\usepackage[
letterpaper,
top=1.2in,
bottom=1.2in,
left=1in,
right=1in]{geometry}
\usepackage{newpxtext} 
\usepackage{textcomp} 
\usepackage[varg,bigdelims]{newpxmath}
\usepackage[scr=rsfso]{mathalfa}
\usepackage{bm} 
\let\mathbb\varmathbb
\usepackage{microtype}
\usepackage%
{hyperref}

\newcommand{\Authornotecolored}[3]{}
\newcommand{\Authorcomment}[2]{}
\newcommand{\Authorfnote}[2]{}

\definecolor{forestgreen(traditional)}{rgb}{0.0, 0.27, 0.13}

\usepackage{boxedminipage}

\newcommand{\abs}[1]{\lvert#1\rvert}

\newcommand{\bigabs}[1]{\big\lvert#1\big\rvert}

\newcommand{\norm}[1]{\lVert#1\rVert}

\newcommand{\Esymb}{\mathbb{E}}
\newcommand{\Psymb}{\mathrm{Pr}}

\DeclareMathOperator*{\E}{\Esymb}

\DeclareMathOperator*{\ProbOp}{\Psymb}
\renewcommand{\Pr}{\ProbOp}

\newcommand{\defeq}{\stackrel{\mathrm{def}}=}

\newcommand\bdot\bullet

\DeclareMathOperator{\poly}{poly}

\newcommand{\Erdos}{Erd\H{o}s\xspace}
\newcommand{\Renyi}{R\'enyi\xspace}
\newcommand{\Lovasz}{Lov\'asz\xspace}

\newcommand{\Bollobas}{Bollob\'as\xspace}

\newcommand{\N}{\mathbb N}
\newcommand{\R}{\mathbb R}

\newcommand{\cP}{\mathcal P}

\newcommand{\cT}{\mathcal T}

\newcommand{\cV}{\mathcal V}

\renewcommand{\S}{\mathbb{S}}

\renewcommand{\leq}{\leqslant}

\renewcommand{\geq}{\geqslant}

\let\epsilon=\varepsilon
\numberwithin{equation}{section}
\newcommand\MYcurrentlabel{xxx}
\newcommand{\MYstore}[2]{%
  \global\expandafter \def \csname MYMEMORY #1 \endcsname{#2}%
}
\newcommand{\MYload}[1]{%
  \csname MYMEMORY #1 \endcsname%
}
\newcommand{\MYnewlabel}[1]{%
  \renewcommand\MYcurrentlabel{#1}%
  \MYoldlabel{#1}%
}
\newcommand{\MYdummylabel}[1]{}
\newcommand{\torestate}[1]{%
  \let\MYoldlabel\label%
  \let\label\MYnewlabel%
  #1%
  \MYstore{\MYcurrentlabel}{#1}%
  \let\label\MYoldlabel%
}
\newcommand{\restatetheorem}[1]{%
  \let\MYoldlabel\label
  \let\label\MYdummylabel
  \begin{theorem*}[Restatement of \cref{#1}]
    \MYload{#1}
  \end{theorem*}
  \let\label\MYoldlabel
}
\newcommand{\restatelemma}[1]{%
  \let\MYoldlabel\label
  \let\label\MYdummylabel
  \begin{lemma*}[Restatement of \cref{#1}]
    \MYload{#1}
  \end{lemma*}
  \let\label\MYoldlabel
}
\newcommand{\restateprop}[1]{%
  \let\MYoldlabel\label
  \let\label\MYdummylabel
  \begin{proposition*}[Restatement of \cref{#1}]
    \MYload{#1}
  \end{proposition*}
  \let\label\MYoldlabel
}
\newcommand{\restatefact}[1]{%
  \let\MYoldlabel\label
  \let\label\MYdummylabel
  \begin{fact*}[Restatement of \prettyref{#1}]
    \MYload{#1}
  \end{fact*}
  \let\label\MYoldlabel
}
\newcommand{\restate}[1]{%
  \let\MYoldlabel\label
  \let\label\MYdummylabel
  \MYload{#1}
  \let\label\MYoldlabel
}

\newcommand{\eps}{\epsilon}

\allowdisplaybreaks
\DeclareMathOperator{\pE}{\tilde{\mathbb{E}}}
\DeclareMathOperator{\Span}{Span}

\def\Span#1{\textbf{Span}\left(#1  \right)}

\def\defeq{\stackrel{\mathrm{def}}{=}}

\def\abs#1{\left|#1  \right|}

\def\norm#1{\left\| #1 \right\|}

\newcommand{\M}{\mathcal{M}}
\newcommand{\cdeg}{\mathrm{cdeg}}
\newcommand{\ip}[1]{\langle #1 \rangle}

\newcommand{\numcolors}{k}
\newcommand{\numlb}{k_0}

\title{A Stress-Free Sum-of-Squares Lower Bound for Coloring}

\author{
   Pravesh K.\ Kothari\thanks{Carnegie Mellon University. Supported by NSF CAREER Award \#2047933.} \\ \texttt{praveshk@cs.cmu.edu} \and Peter Manohar\thanks{Carnegie Mellon University. Supported by the NSF Graduate Research Fellowship Program
and the ARCS Foundation. This material is based upon work supported by the National Science Foundation Graduate Research Fellowship Program under Grant No.\ DGE 1745016. Any opinions, findings, and conclusions or recommendations expressed in this material are those of the author(s) and do not necessarily reflect the views of the National Science Foundation.} \\ \texttt{pmanohar@cs.cmu.edu} }%

\date{\today}

\usepackage{empheq}

\begin{document}

\pagestyle{empty}


\maketitle
\thispagestyle{empty} 


\begin{abstract}
We prove that with high probability over the choice of a random graph $G$ from the \Erdos-\Renyi distribution $G(n,1/2)$, a natural $n^{O(\epsilon^2 \log n)}$-time, degree $O(\epsilon^2 \log n)$ sum-of-squares semidefinite program cannot refute the existence of a valid $k$-coloring of $G$ for $k = n^{1/2 +\epsilon}$. Our result implies that the refutation guarantee of the basic semidefinite program (a close variant of the \Lovasz theta function) cannot be appreciably improved by a natural $o(\log n)$-degree sum-of-squares strengthening, and this is tight up to a $n^{o(1)}$ slack in $k$. To the best of our knowledge, this is the first lower bound for coloring $G(n,1/2)$ for even a single round strengthening of the basic SDP in any SDP hierarchy. 

Our proof relies on a new variant of instance-preserving \emph{non-pointwise complete reduction} within SoS from coloring a graph to finding large independent sets in it. Our proof is (perhaps surprisingly) short, simple and \emph{does not} require complicated spectral norm bounds on random matrices with dependent entries that have been otherwise necessary in the proofs of many similar results~\cite{DBLP:conf/focs/BarakHKKMP16,8104104,DBLP:journals/corr/abs-1907-11686,ghosh2020sumofsquares,10.1145/3357713.3384319}. 

Our result formally holds for a constraint system where vertices are allowed to belong to multiple color classes; we leave the extension to the formally stronger formulation of coloring, where vertices must belong to unique colors classes, as an outstanding open problem. 

\end{abstract}

\clearpage


  \microtypesetup{protrusion=false}
  \tableofcontents{}
  \microtypesetup{protrusion=true}

\clearpage

\pagestyle{plain}
\setcounter{page}{1}
\renewcommand{\pE}{\tilde{\mathbb{E}}}
\section{Introduction}

Starting with the seminal work of Arora, \Bollobas, \Lovasz and Tourlakis ~\cite{MR2322869-Arora06}, understanding the power of systematic hierarchies of linear and semidefinite programs for solving combinatorial optimization problems has been a foundational goal in complexity theory. This project has achieved many successes including sharp lower bounds for basic problems~\cite{MR2402454-Schoenebeck07,MR2780074-Charikar09,MR2648454-Raghavendra09,DBLP:conf/ipco/BenabbasM10,DBLP:conf/fsttcs/BenabbasCGM11,DBLP:journals/toc/BenabbasGMT12,DBLP:conf/coco/GhoshT17} in various hierarchies of linear and semidefinite programs~\cite{MR1806434-Lasserre00,parrilo2000structured, MR1061981-Sherali90, DBLP:journals/siamjo/LovaszS91} (see~\cite{MR2894694-Chlamtac12,TCS-086} for expositions).

However, proving lower bounds for the sum-of-squares (SoS) semidefinite programming hierarchy -- the strongest known hierarchy of efficiently solvable convex programs -- has achieved only a limited amount of success.  This is partially explained by the remarkable success of the SoS hierarchy in designing state-of-the-art algorithms for worst-case optimization problems such as max-cut~\cite{MR1412228-Goemans95}, sparsest cut~\cite{MR2535878-Arora09}, unique games on general~\cite{MR3424199-Arora15} and algebraic graphs~\cite{DBLP:journals/corr/abs-2006-09969,DBLP:journals/corr/abs-2011-04658},  quantum separability~\cite{DBLP:conf/stoc/BarakKS17} and more recently, a string of successes in high-dimensional algorithmic statistics including robust estimation of moments~\cite{DBLP:journals/corr/abs-1711-11581}, clustering spherical~\cite{HopkinsLi17,KothariSteinhardt17} and non-spherical mixture models~\cite{bakshi2020outlierrobust, DHKK20}, robust learning of all Gaussian mixtures~\cite{DBLP:journals/corr/abs-2012-02119, DBLP:journals/corr/abs-2011-03622}, list-decodable learning~\cite{DBLP:journals/corr/abs-1905-05679,DBLP:journals/corr/abs-1905-04660,DBLP:journals/corr/abs-2002-05139,raghavendra2020list}, tensor decomposition~\cite{DBLP:conf/focs/MaSS16}, and sparse~\cite{DBLP:conf/focs/dOrsiKNS20} and tensor principal component analysis~\cite{DBLP:conf/colt/HopkinsSS15}, among others. Indeed, given the remarkable power of the SoS method in designing algorithms for such \emph{average-case} settings, SoS lower bounds (and related restricted algorithmic techniques such as the low-degree polynomial method~\cite{8104104,DBLP:conf/focs/HopkinsS17,DBLP:journals/corr/abs-1907-11636}) are increasingly used to ascertain average-case hardness and \emph{algorithmic thresholds}. 

In the last few years, there has been some progress in proving sum-of-squares lower bounds for average-case problems~\cite{MR1880908-Grigoriev01,MR1875740-Grigoriev02,DBLP:conf/focs/Schoenebeck08,MR2780076-Tulsiani09,MR3388187-Barak15,DBLP:conf/focs/BarakHKKMP16,DBLP:conf/stoc/KothariMOW17,ghosh2020sumofsquares,10.1145/3357713.3384319}. However, such progress has come about via fairly technical\footnote{Almost all recent analyses run into $\sim$ 50 pages!}, problem-specific arguments and a host of natural questions, e.g.\ combinatorial optimization on sparse random graphs, remain out of reach of current techniques. In particular, a central challenge in this line of work has been to analyze the sum-of-squares semidefinite programs for refuting the existence of a $\numcolors$-coloring in \Erdos-\Renyi random graphs. Classical works~\cite{MR951992} in probability showed that the chromatic number of $G \sim G(n,1/2)$\footnote{Recall $G \sim G(n,1/2)$ is a graph on $n$ vertices where each edge $\{i,j\}$ is independently included with probability $1/2$.} is tightly concentrated around $n/2 \log_2 n$. However, the best known polynomial time algorithm (corresponding to the degree 2 SoS relaxation, a close variant of the famous \Lovasz theta function) can only \emph{refute} the existence of a $\sqrt{n}$-coloring in such random graphs\footnote{We note that a close variant of \Lovasz-theta function is also a crucial component in the current state-of-the-art algorithms for worst-case coloring of $k$-colorable graphs with a small polynomial number of colors~\cite{DBLP:journals/jacm/KargerMS98,MR2277147-Arora06,DBLP:conf/focs/Chlamtac07}.}. While it is natural to guess that higher-degree relaxations yield no significant improvement, establishing this has proved to be an elusive goal. Indeed, even the easier goal of establishing sharp SoS lower bounds for the clique number of $G \sim G(n,1/2)$ required~\cite{DBLP:conf/focs/BarakHKKMP16} the introduction of \emph{pseudo-calibration} -- a technique that has found several further uses in establishing SoS lower bounds for average-case problems. However, analyzing lower bound constructions based on pseudo-calibration requires understanding the spectra of complicated random matrices with dependent random entries. While this has been accomplished for a few select examples~\cite{8104104,ghosh2020sumofsquares}, the case of graph coloring seems to be particularly unwieldy and has thus resisted progress so far. 

In this paper, we establish a tight SoS lower bound for a natural higher-degree SoS relaxation of the graph coloring problem in $G(n,1/2)$. Our proof circumvents pseudo-calibration entirely. Instead, we exhibit a \emph{non-pointwise complete reduction} -- a notion of reductions that departs from the standard framework introduced by Tulsiani~\cite{MR2780076-Tulsiani09} (and used in \cite{DBLP:conf/soda/BhaskaraCVGZ12}) -- that  obtains a lower bound for the coloring problem from a lower bound for the independent set problem (see Section~\cref{sec:tul} for a detailed discussion). Somewhat surprisingly, our analysis does \emph{not} require spectral analysis of complicated random matrices and instead succeeds whenever the lower bound construction for the independent set problem satisfies some natural covering properties. Our main result then follows by verifying these properties for the construction of \cite{DBLP:conf/focs/BarakHKKMP16}.

\subsection{Results}
Our results apply to the following polynomial constraint system in the real-valued variables $\{x_{i,c}\}_{i \in[n], c \in [k]}$ that is satisfiable if and only if the graph $G$ is $k$-colorable.

\begin{center}
\begin{subequations}
\label[colorconstraints]{colorconstraints}
\boxed{
\begin{aligned}
\textrm{\colorconstraints} \\
x_{i,c}^2 = x_{i,c} &\text{ for all } i \in [n], c \in [k] & \ \text{(Booleanity Constraints)} \\
x_{i,c} x_{j,c} = 0 &\text{ for all } c\in [k] \text{ and } \{i,j\} \in E(G) & \ \text{(Edge Constraints)}  \\
\sum_{c} x_{i,c} \geq 1 &\text{ for all } i \in [n] & \ \text{(Sum Constraints)}
\end{aligned}
}
\end{subequations}
\end{center}
 In \cref{colorconstraints}, the variable $x_{i,c}$ represents the $0$-$1$ indicator of whether the $i$ vertex is in the $c$-th color class.
 The booleanity constraints enforce that $x_{i,c} \in \{0,1\}$, the edge constraints enforce that if $\{i,j\} \in E(G)$, then the subset of colors assigned to $i$ is disjoint from the subset of colors assigned to $j$, and the sum constraints enforce that each vertex is in at least one color class.

\cref{colorconstraints} allow for a vertex to be in \emph{more than one} color class. Our lower bound technique does not currently succeed for the related set of constraints where each vertex must belong to exactly one color class. See \cref{sec:weak-vs-strong} for a discussion on the difference between the formulations.

Our main result shows that with high probability over the draw of $G \sim G(n,1/2)$, the degree $O(\eps^2 \log n)$ SoS proof system cannot refute \cref{colorconstraints} for $G$ when $k = n^{\frac{1}{2} + \eps}$. 

\begin{mtheorem}
\label{thm:coloringlb}
Let $n$ be sufficiently large positive integer and $\eps \in (\Omega(\sqrt{\frac{1}{\log n}}), \frac{1}{2})$. Then, for $k = n^{\frac{1}{2} + \eps}$ and $d = O(\eps^2 \log n)$, with probability $1 - 1/\poly(n)$ over the draw of $G \sim G(n,1/2)$, the $n^{O(d)}$-time, degree $d$ sum-of-squares relaxation of \cref{colorconstraints} cannot refute the existence of a $k$-coloring of $G$. 
\end{mtheorem}
Equivalently, \cref{thm:coloringlb} says that with high probability over $G \sim G(n,1/2)$, \cref{colorconstraints} do not admit an $O(\eps^2 \log n)$-degree positivstellensatz refutation when $k = n^{\frac{1}{2}+\eps}$. As was formally verified in \cite{DBLP:journals/siamcomp/BanksKM19}\footnote{\cite{DBLP:journals/siamcomp/BanksKM19} proved this equivalence for a slightly different formulation of \cref{colorconstraints}, which we will discuss in \cref{sec:weak-vs-strong}. However, the same proof works even for \cref{colorconstraints}.}, a degree $2$ coloring pseudo-expectation is equivalent to a vector solution with value at least $k$ to the semidefinite program that computes the \Lovasz theta function. To the best of our knowledge, this result gives the first lower bound for $\omega(1)$ rounds (or even a single round of strengthening of the basic SDP) in a natural SDP hierarchy.
 
\begin{remark}[Tightness of \Cref{thm:coloringlb}] It is well-known~\cite{DBLP:journals/cpc/Coja-Oghlan05,DBLP:conf/approx/BanksKM17} that the degree 2 sum-of-squares relaxation of \cref{colorconstraints} can refute the existence of $k$-coloring in $G \sim G(n,1/2)$ for $k = O(\sqrt{n})$. Thus, our lower bound in \Cref{thm:coloringlb} is tight up to a $n^{\eps}$ factor in $k$. On the other hand, we give a simple proof in \cref{sec:independentsetrefutation} that shows that the degree $8(1 + o(1)) \log_2 n$ SoS relaxation of \cref{colorconstraints} succeeds in refuting the existence of a $k$-coloring in $G(n,1/2)$ (w.h.p.) for the nearly optimal~\cite{MR951992} bound of $k \leq \frac{n}{e \cdot 2(1 + o(1))\log_2 n}$. Hence, the upper bound on $d$ in \cref{thm:coloringlb} is tight up to constants.
\end{remark}

\subsection{A non-pointwise complete SoS reduction from coloring to independent set} 
Using standard SDP duality, proving \cref{thm:coloringlb} is equivalent to proving the existence of a dual witness called a pseudo-expectation defined below (see lecture notes~\cite{BarakS16} and the monograph~\cite{TCS-086} for background). 

\begin{definition}[Pseudo-expectation for Coloring]
A degree $d$ coloring pseudo-expectation $\pE$ for $G$ using $\numcolors$ colors is a linear operator that maps polynomials of degree $\leq d$ in variables $\{x_{i,c}\}_{ i \in [n], c \in [k]}$ to $\R$, satisfying the following three properties:
\begin{enumerate}
	\item \textbf{Normalization: } $\pE[1]  = 1$,
	\item \textbf{Positivity: } $\pE[f^2] \geq 0$ for every polynomial $f$ of degree at most $d/2$,
	\item \textbf{Coloring Constraints}: $\pE$ satisfies \cref{colorconstraints}. \begin{enumerate}[(a)] 
	\item for every polynomial $f$ of degree at most $d-2$, $\pE[f \cdot (x_{i,c}^2 - x_{i,c})] = 0$,
	\item for every polynomial $f$ of degree at most $d-2$ and any edge $\{i, j\} \in E(G)$, $\pE[f \cdot x_{i,c} x_{j,c}] = 0$,
	 \item for every polynomial $f$ of degree at most $\frac{d-1}{2}$, $\pE[f^2 \cdot (\sum_{c \leq k} x_{i,c} -1)] \geq 0$. 
	\end{enumerate}
	\end{enumerate}
\end{definition}
In order to prove \cref{thm:coloringlb}, it suffices to show that with high probability over the draw of $G \sim G(n,1/2)$, there is a degree $O(\eps^2 \log n)$ coloring pseudo-expectation for the graph $G$ that uses $\numcolors = n^{\frac{1}{2}+ \eps}$ colors. Somewhat surprisingly, we prove the existence of such a pseudo-expectation essentially without any random matrix analysis. Instead, we construct a coloring pseudo-expectation $\pE'$ for $G$ from a pseudo-expectation $\pE$ satisfying the related independent set constraints for \emph{the same graph $G$} whenever $\pE$ satisfies two additional natural ``covering'' properties. We recall the definition of an independent set pseudo-expectation below.

\begin{definition}[Pseudo-expectation for Independent Set]
A degree $d$ independent set pseudo-expectation $\pE$ is a linear operator that maps polynomials of degree $\leq d$ in variables $\{x_i\}_{i \in[n]}$ to $\R$, satisfying the following three properties:
\begin{enumerate}
	\item \textbf{Normalization: } $\pE[1]  = 1$,
	\item \textbf{Positivity: } $\pE[f^2] \geq 0$ for every polynomial $f$ of degree at most $d/2$,
	\item \textbf{Independent Set Constraints:} For every polynomial $f$ of degree at most $d-2$, $\pE[f \cdot (x_{i}^2 - x_{i})] = 0$ and $\pE[f \cdot x_{i} x_{j}] = 0$ for any edge $\{i,j\} \in E(G)$.
	\end{enumerate}
\end{definition}

Our main result that constructs a reduction from coloring to independent set is described below.
\begin{mtheorem}
\label{thm:reduction}
Let $G$ be a graph on $n$ vertices, and let $\pE$ be a degree $d$ independent set pseudo-expectation. Suppose further that $\pE$ satisfies the two ``covering'' properties: \begin{inparaenum}[(1)] \item $\pE[x_i] \geq \frac{1}{\numlb}$ for some integer $\numlb$, and \item there exists $\lambda \in \R_{> 0}$ such that for all multilinear $f$ with $\deg(f) \leq d/2$, $\pE[f^2] \geq \lambda \norm{\Pi_G f}_2^2$, where $\Pi_G$ is the projection of $f$ onto the linear subspace orthogonal to $\{g x_i x_j : \{i,j\} \in E(G), \deg(g) \leq d-2\}$ (viewed as a subset of coefficient vectors of polynomials with the Euclidean inner product), and $\norm{f}_2$ denotes the $\ell_2$ norm of the polynomial of $f$, viewed as a coefficient vector\end{inparaenum}. 
Then, there is a degree $d' := 1 + d/2$ coloring pseudo-expectation $\pE'$ using $\numcolors = O(\numlb d \log (n^d/\lambda))$ colors.
\end{mtheorem}
\cref{thm:coloringlb} follows by verifying (see \cref{sec:fmain}) that the independent set pseudo-expectation constructed in~\cite{DBLP:conf/focs/BarakHKKMP16} satisfies the hypotheses of \cref{thm:reduction} with $\numlb = n^{\frac{1}{2} + \eps}$ and $\lambda = n^{-O(d)}$.

\cref{thm:reduction} holds for \emph{every} graph $G$ that admits an independent set pseudo-expectation satisfying the two additional covering properties. Hence, \Cref{thm:reduction} gives a reduction ``within SoS'' from the problem of coloring $G$ to the problem of finding a large independent set in $G$. As a consequence of the modularity of \cref{thm:reduction}, we have also reduced the task of proving SoS lower bounds for coloring for $G(n,p)$ with $p \ll 1/2$ to the task of ``merely'' proving a similar lower bound for independent set for $G(n,p)$. The latter task, though challenging, appears significantly less daunting than attacking coloring directly. 

To understand the two covering properties intuitively, note that even in ``real-life'' the existence of a single large independent set (say of size $\sim n/k$) does not imply the existence of a $k$-coloring of $G$. However, the existence of a $k$-coloring follows if we can prove that there is a collection of $\numcolors$ independent sets that \emph{cover} all vertices of $G$. The conditions appearing in \cref{thm:reduction} can be thought of as forcing two ``low-degree'' consequences of such a uniform covering property on the pseudo-expectation for independent sets. Informally, the first constraint says that each vertex $i$ appears in the independent set with reasonable probability, and the second constraint says that the minimum eigenvalue of $\pE$ is not too small, once we ignore polynomials that are required to have pseudo-expectation $0$ due to the independent set constraints.

\subsection{Comparison with Tulsiani's framework} 
\label{sec:tul}

Our proof of \cref{thm:reduction} requires a notion of reduction that departs from the standard framework introduced in~\cite{DBLP:conf/stoc/Tulsiani09}. Tulsiani's method\footnote{What follows is an equivalent description of Tulsiani's work in the language of pseudo-expectations.} uses a \emph{pointwise complete} reduction from problem $B$ to problem $A$ to construct a pseudo-expectation consistent with a polynomial formulation for $B$ from a pseudo-expectation consistent with a polynomial formulation for $A$. Specifically, a \emph{pointwise complete} SoS reduction from problem $B$ to problem $A$ is a map from instances $I_A$ of problem $A$ to instances $I_B$ of problem $B$, along with a ``solution map'' $x \mapsto y$ 
that takes any solution $x$ of instance $I_A$ into a solution $y$ of instance $I_B$ that, in addition, satisfies: \begin{inparaenum}[(1)] \item each entry of the solution map $x\rightarrow y$ is computable by low-degree polynomials, and \item there is a ``low-degree sum-of-squares proof'' that if $x$ satisfies the constraint system $A$ for instance $I_A$ then $y$ satisfies the constraint system $B$ for instance $I_B$ \end{inparaenum}. In particular, if $y_i = p_i(x)$ for each $i$ for polynomials $p_1, p_2, \ldots$ of degree most $d_1$, then the framework allows us to transform a degree $d$ pseudo-expectation consistent with $A$ into a degree $\approx d/d_1$ pseudo-expectation consistent with $B$.  Tulsiani used this machinery to prove several SoS lower bounds for \emph{worst-case} combinatorial optimization problems such as constraint satisfaction, vertex cover, independent set and coloring. 

In average-case settings, however, we need tight control over the map between instances $I_A$ and $I_B$ in order to obtain a lower bound that applies to the target distribution over the instances of problem $B$. This makes Tulsiani's method not directly applicable to our setting since (if we insist on instance-preserving reductions) there is provably no pointwise complete, instance-preserving reduction from $k$-coloring to independent set. This is because the existence of a large independent set in $G$ does not, in general, imply the existence of a valid coloring of $G$ with a small number of colors. Instead, as we discuss next, our reduction directly maps a pseudo-expectation for independent set into a pseudo-expectation for coloring as long as the pseudo-expectation for independent set satisfies the additional uniform covering conditions.  

\subsection{Proof overview: coloring by repeated sampling} 
We describe our construction and a couple of main insights that go into the proof of \cref{thm:reduction} here. These ideas make the proof of \cref{thm:coloringlb} ``stress-free'': they allow us to completely sidestep the technical complexity of analyzing constructions based on pseudo-calibration that involve computing the spectra of certain random matrices (called graphical matrices) for proving SoS lower bounds. 

 We begin by describing the conceptual heart of the idea. In order to do this, it is helpful to consider the thought experiment (and very special case!) where the independent set pseudo-expectation $\pE$ is in fact the expectation operator  $\E_{\mu}$ associated with some \emph{distribution} $\mu$ on independent sets of $G$. Further, suppose that $\E_{\mu}[x_i] = \Pr[i \in S] \geq \frac{1}{\numlb}$. Then, observe that we can immediately derive that $G$ must be $k$-colorable with $O(\numlb \log n)$ colors. In fact, we can produce a simple, explicit probability distribution on $k$-colorings of $G$: independently sample $\numcolors$ independent sets $S_1, \dots, S_{\numcolors}$ from $\mu$ and set each of them to be a new color class. Observe that the chance that a certain vertex is not included in any of the $S_i$'s is $\Pr[i \notin \cup_{j = 1}^{\numcolors} S_j] = (1 - \frac{1}{\numlb})^{\numcolors} \leq e^{-\numcolors/\numlb} \ll \frac{1}{n}$ for $\numcolors = O(\numlb \log n)$, and hence by a union bound, using the $S_i$'s as color classes gives a valid $k$-coloring of $G$ with high probability.\footnote{Note that a vertex $i$ will, with high probability, belong to multiple color classes. In order to obtain a valid $\numcolors$-coloring, we simply remove each vertex from all but one of its assigned color classes.} To get a distribution $\mu'$ entirely supported over $\numcolors$-colorings of $G$, one can simply sample $S_1, \dots, S_{\numcolors}$ from $\mu$ \emph{conditioned} on the high probability event that $\cup_{j = 1}^{\numcolors} S_j = V(G)$. 

Our key idea is to replicate this ``independent sampling'' step within the sum-of-squares framework. For pseudo-expectations, independent sampling produces a pseudo-expectation on a tuple of $\numcolors$ independent sets given by the $\numcolors$-th tensor $\pE^{\otimes \numcolors}$. However, there is no natural way to perform the final ``conditioning'' step for low-degree pseudo-expectations\footnote{There is a natural and standard way to import ``conditioning'' of probability distributions into the SoS framework via ``polynomial reweightings'' (see ~\cite{DBLP:conf/stoc/BarakKS17} for a formal treatment of such reweightings). However, the relevant polynomial $\prod_{i} (1 - \prod_{c} (1 - x_{i,c}))$ in our case has degree $n \numcolors$, and so we would need the independent set pseudo-expectation to have degree $\gg n$ in order for the reweighting to be well-defined!}, which for distributions is the simplest way to ensure the ``covering property'', that is,  $i \in \cup_{j = 1}^{\numcolors} S_j$ for every $i$, or equivalently to make $\E_{\mu'}$ satisfy the sum constraints $\sum_c x_{i,c} \geq 1$ for every $i$. 

The sampling analogy suggests a way out, however: observe that when one draws $S_1, \dots, S_{\numcolors}$ from an actual probability distribution $\mu$ on independent sets that satisfies $\E_{\mu} x_i \geq 1/\numlb$, we expect each $i$ to be in not just one but in fact in $\frac{k}{\numlb} = O(\log n)$ of the subsets. Equivalently, we expect that $\E_{\mu} \sum_c x_{i,c} = \Omega(\log n)$. Because this expectation is large, if low-degree polynomials of $\mu$ are sufficiently well-concentrated around their expectations, we may expect that the influence of the points $x$ in the support of $\mu$ where $\sum_{c} x_{i,c} \leq 1 \ll \E_{\mu} \sum_{c} x_{i,c}$ to be small. Thus, one may hope that expectations of low-degree ($\deg \leq d$) polynomials cannot ``distinguish'' between distributions $\mu$ where every point in the support of $\mu$ satisfies $\sum_c x_{i,c} > 1$ versus those where the probability of $\sum_c x_{i,c} = 0$ is non-zero for some $i$. In that case, one might expect $\pE^{\otimes k}$ to satisfy the sum constraints. 

Our actual proof establishes precisely such a statement even for pseudo-distributions whenever the smallest nontrivial eigenvalue of the pseudo-moment matrix of the independent set pseudo-expectation $\pE$ is not too small. We show that this condition implies a non-trivial eigenvalue lower bound for the $k$-fold tensor power of $\pE$ on polynomials of total degree\footnote{Notice that $\pE^{\otimes \numcolors}$ is defined and even positive semidefinite on a \emph{larger} subspace of polynomials that includes some of total degree $\sim kd$!} $d$. A direct argument relying on spectra of the tensor product of matrices yields an estimate that decays exponentially in $k$, which is too weak for us. Instead, we show that the smallest eigenvalue of $\pE^{\otimes k}$ \emph{when restricted to the subspace of polynomials of total degree $\leq d$} decays only as an exponential in $d \log n$. While eventually elementary, this argument is both crucial and somewhat technical and is presented in full in~\cref{sec:eiglb}. 

Intuitively, a good enough lower bound on the smallest non-zero eigenvalue of $\pE^{\otimes k}$ on the relevant subspace of polynomials is our ``surrogate'' for the concentration of low-degree polynomials that we needed in the case of actual probability distributions above. Concretely, we use this non-trivial eigenvalue lower bound on $\pE^{\otimes k}$ as follows: let $h_i$ be the indicator polynomial of the ``bad event'' $\sum_c x_{i,c} \leq 1$. Then, we prove that for a polynomial $f$ to be able to ``detect'' this event, we must have $\pE^{\otimes \numcolors}[f^2 h_i] = \Omega(\pE^{\otimes \numcolors}[f^2])$. However, applying Cauchy-Schwarz, we have that $\pE^{\otimes \numcolors}[f^2 h_i] \leq \sqrt{\pE^{\otimes \numcolors}[f^4] \pE^{\otimes \numcolors}[h_i^2]}$. We show that the smallest eigenvalue condition implies a \emph{$2 \shortrightarrow 4$ hypercontractive inequality} on the pseudo-expectation operator on polynomials of total degree $\leq d$, i.e., $\pE^{\otimes \numcolors}[f^4]  \leq (\frac{n^{O(d)}}{\lambda})^d \pE^{\otimes \numcolors}[f^2]^2$. Combined with the estimate (that one roughly expects to hold from the independent sampling based argument) $\pE^{\otimes \numcolors}[h_i] \approx e^{-\numcolors/\numlb}$, this yields that $\pE^{\otimes \numcolors}$ indeed satisfies the constraints $\sum_{c} x_{i,c} \geq 1$ for every $i$, when $\numcolors = O(\numlb d \log (n^d/\lambda))$.

\subsection{Weak vs.\ strong formulation for coloring} \label{sec:weak-vs-strong}
The coloring axioms are often stated with an equality (we call this the \emph{strong form}) in the sum constraints along with the additional constraints $\{x_{i,c} x_{i,c'} = 0 : c \ne c'\}$, instead of an inequality (the \emph{weak form}) as done in \cref{colorconstraints}. Namely, the strong coloring constraints are the following.
\begin{center}
\begin{subequations}
\boxed{
\begin{aligned}
\text{(Strong) Color Constraints} \\
x_{i,c}^2 = x_{i,c} &\text{ for all } i \in [n], c \in [k] & \ \text{(Booleanity Constraints)} \\
x_{i,c} x_{j,c} = 0 &\text{ for all } c\in [k] \text{ and } \{i,j\} \in E(G) & \ \text{(Edge Constraints)}  \\
\sum_{c} x_{i,c} = 1 &\text{ for all } i \in [n] & \ \text{(Sum Equality Constraints)} \\
x_{i,c} x_{i,c'} = 0 &\text{ for all } c \ne c' \in [k] & \ \text{(Same Color Constraints)}
\end{aligned}
}
\end{subequations}
\end{center}
When viewed as a polynomial optimization problem, there is no difference between the weak and strong formulations: one is satisfiable if and only if the other is. Further, SoS relaxations of both formulations ``converge'' (i.e., refute $k$-coloring in $G(n,1/2)$ for the right value of $k$) at $O(\log n)$ degree, and both imply corresponding lower bounds for independent set: a degree $d$ coloring (weak or strong) pseudo-expectation with $\numcolors$ colors can easily be transformed into a degree $d$ independent set pseudo-expectation with independent set size $\geq \frac{n}{\numcolors}$. Thus, while the SoS relaxation of the strong form is \emph{formally} stronger (for degrees $>2$), \cref{colorconstraints} do not appear to meaningfully weaken the strong formulation. 

However, at the moment our technique does not succeed in constructing a pseudo-expectation that satisfies the constraints in the strong formulation. This is an important technical issue encountered in proving several prior SoS lower bounds where it turns out to be unwieldy to handle ``hard'' constraints such as those formulated by an exact polynomial equality. For example, in the planted clique problem, one may naturally wish for the pseudo-expectation to satisfy the clique-size constraint ``$\sum_i x_i = \omega$'' \emph{exactly}. While this is achieved for the degree $4$ pseudo-expectation of~\cite{DBLP:journals/corr/HopkinsKP15,DBLP:journals/corr/RaghavendraS15}, the degree $\sim \log n$ pseudo-expectation constructed in~\cite{DBLP:conf/focs/BarakHKKMP16} does \emph{not} satisfy this as a constraint. This technical deficiency can sometimes even be crucial in downstream applications. For example, the construction of the hardness result for finding Nash equilibria in two player games in~\cite{DBLP:conf/stoc/KothariM18} (see also the discussion in~\cite{DBLP:journals/corr/abs-1809-01207}) needs elaborate work-arounds in order to work without satisfying such exact constraints. 

Informally speaking, this is because the proofs of positivity of candidate pseudo-expectations rely on ``collecting terms'' together in the \emph{graphical matrix} (a class of structured random matrices) decomposition in order to form PSD matrices. This aggregation step needs coefficients on various graphical matrices appearing in the decomposition to satisfy certain exact relationships. Modifying such coefficients to satisfy hard constraints while maintaining positivity appears challenging. 

Such technical difficulty has been dealt with in some special cases (where the analyses did not need pseudo-calibration in the first place). For example, for the (much simpler) case of constraint satisfaction problems with a single global equality constraint, this problem was addressed via certain \emph{ad hoc} methods in a recent work~\cite{DBLP:journals/corr/abs-1809-01207}. That work also includes a longer discussion on the issues arising in constructing pseudo-expectations satisfying hard constraints. Finding general techniques to design and analyze pseudo-expectations that exactly satisfy multiple hard constraints \emph{simultaneously} -- such as those arising in the strong formulation of graph coloring -- is an important and challenging open problem.


\section{Reduction to SoS Lower Bounds for Independent Set}
\label{sec:mainanalysis}

In this section, we prove \cref{thm:reduction} (restated below).

\begin{theorem*}[\Cref{thm:reduction}, restated]
Let $G$ be a graph on $n$ vertices, and let $\pE$ be a degree $d$ independent set pseudo-expectation. Suppose further that $\pE$ satisfies the two ``covering'' properties: \begin{inparaenum}[(1)] \item $\pE[x_i] \geq \frac{1}{\numlb}$ for some integer $\numlb$, and \item there exists $\lambda \in \R_{> 0}$ such that for all multilinear $f$ with $\deg(f) \leq d/2$, $\pE[f^2] \geq \lambda \norm{\Pi_G f}_2^2$, where $\Pi_G$ is the projection of $f$ onto the linear subspace orthogonal to $\{g x_i x_j : \{i,j\} \in E(G), \deg(g) \leq d-2\}$ (viewed as a subset of coefficient vectors of polynomials with the Euclidean inner product), and $\norm{f}_2$ denotes the $\ell_2$ norm of the polynomial of $f$, viewed as a coefficient vector\end{inparaenum}. 
Then, there is a degree $d' := 1 + d/2$ coloring pseudo-expectation $\pE'$ using $\numcolors = O(\numlb d \log (n^d/\lambda))$ colors.
\end{theorem*}

\subsection{\emph{Coloring degree} of polynomials}
Before proceeding with the proof, we first introduce some notation. Let $f$ be a polynomial in the variables $\{x_{i,c}\}_{i \in [n], c \in [\numcolors]}$. We define the \emph{coloring degree} of $f$, denoted by $\cdeg(f)$, to be the maximum, taken over all monomials $\prod_{c = 1}^{\numcolors} \prod_{i = 1}^n x_{i,c}^{\alpha_{i,c}}$ for which $f$ has a nonzero coefficient, of $\max_{c \in [\numcolors]} \deg(\prod_{i = 1}^n x_{i,c}^{\alpha_{i,c}})$. As an example, the polynomial $x_{1,1} x_{1,2}$ has degree $2$ and coloring degree $1$, while the polynomial $x_{1,1} x_{2,1}$ has degree $2$ and coloring degree $2$. Informally, the coloring degree only ``charges'' a polynomial for degrees in variables of a single color. 

Let $\cP_d$ denote the set of polynomials in the variables $\{x_i\}_{i \in [n]}$ of degree at most $d$. Then, the set of coloring degree $\leq d$ polynomials is precisely $\cP_d^{\otimes \numcolors}$. Recall that the operator $\Pi_G$ is the projection of $f \in \cP_d$ to the subspace orthogonal to $\{g x_i x_j : \{i,j\} \in E(G), \deg(g) \leq d-2\}$. We let $\Pi_G^{\otimes \numcolors}$ denote the $\numcolors$-th tensor of $\Pi_G$. Namely, $\Pi_G^{\otimes \numcolors}$ is the projection of $f \in \cP_d^{\otimes \numcolors}$ to the subspace orthogonal to $\{g x_{i,c} x_{j,c} : \{i,j\} \in E(G), c \in [\numcolors], \cdeg(g x_{i,c} x_{j,c}) \leq d\}$. 

Recall that a degree $d$ pseudo-expectation is a linear operator $\pE \colon \cP_d \to \R$ such that $\pE[1] = 1$, and $\pE[f^2] \geq 0$ for all $f$ with $\deg(f) \leq d/2$. For a pseudo-expectation $\pE \colon \cP_d \to \R$, we define $\pE^{\otimes \numcolors} \colon \cP_d^{\otimes \numcolors} \to \R$ to be the $\numcolors$-th tensor of $\pE$. Concretely, $\pE^{\otimes \numcolors}$ is a pseudo-expectation in the variables $\{x_{i,c}\}_{i \in [n], c \in [\numcolors]}$, defined as follows. For polynomials $f_1, \dots, f_{\numcolors}$ where \begin{inparaenum}[(1)] \item $f_c$ is a polynomial in the variables $\{x_{i,c}\}_{i \in [n]}$ for each $c$, and \item $\deg(f_c) \leq d$ for all $c$\end{inparaenum}, we first define $\pE^{\otimes \numcolors}[\prod_{c = 1}^{\numcolors} f_c] \defeq \prod_{c = 1}^{\numcolors} \pE[f_c]$, and then extend $\pE^{\otimes \numcolors}$ to be defined on all $f \in \cP_d^{\otimes \numcolors}$ via linearity. We define a \emph{coloring degree $d$ pseudo-expectation} to be a linear operator $\pE \colon \cP_d^{\otimes \numcolors} \to \R$ such that $\pE[1] = 1$ and $\pE[f^2] \geq 0$ for all $f$ with $\cdeg(f) \leq d/2$. If $\pE$ is a degree $d$ pseudo-expectation, then $\pE^{\otimes \numcolors}$ is a coloring degree $d$ pseudo-expectation.

It is well-known that degree $d$ pseudo-expectations satisfy the Cauchy-Schwarz inequality:
\begin{fact}[See~\cite{BarakS16}]
\label{fact:cauchyschwarz}
Let $f,g$ be polynomials with $\deg(f), \deg(g) \leq d/2$, and let $\pE$ be a degree $d$ pseudo-expectation. Then $\pE[f g] \leq \sqrt{\pE[f^2] \pE[g^2]}$.
\end{fact}
We observe that a similar fact also holds for coloring degree $d$ pseudo-expectations.
\begin{fact}
\label{fact:cauchyschwarzcdeg}
Let $f,g$ be polynomials with $\cdeg(f), \cdeg(g) \leq d/2$, and let $\pE$ be a coloring degree $d$ pseudo-expectation. Then $\pE[f g] \leq \sqrt{\pE[f^2] \pE[g^2]}$.
\end{fact}
The proof of \cref{fact:cauchyschwarzcdeg} is nearly identical to the proof of \cref{fact:cauchyschwarz}, as the proof of \cref{fact:cauchyschwarz} merely requires that $\pE$ is a pseudo-expectation where $\pE[f^2]$, $\pE[g^2]$, $\pE[fg]$ and $\pE[(f-g)^2]$ are all well-defined.

\subsection{Proof of \cref{thm:reduction}}
\parhead{Construction of the pseudo-expectation}
Fix a graph $G$ and degree bound $d$, and let $\pE \colon \cP_d \to \R$ be a degree $d$ independent set pseudo-expectation for $G$ such that \begin{inparaenum}[(1)] \item $\pE[x_i] \geq \frac{1}{\numlb}$, and \item $\pE[f^2] \geq \lambda \norm{\Pi_G f}_2^2$\end{inparaenum}. Let $\numcolors \in \N$ to be chosen later, and assume without loss of generality that $\numcolors < n$.  

Let $\pE^{\otimes \numcolors} \colon \cP_d^{\otimes \numcolors} \to \R$ be the $\numcolors$-fold tensor power of $\pE$, and let $\pE'$ be the pseudo-expectation defined over all polynomials $f$ that have \emph{degree} at most $1 + d/2$ obtained by restricting $\pE^{\otimes \numcolors}$ to this subspace. 

\parhead{Analysis of the constraints} 
We first observe that both $\pE'$ and $\pE^{\otimes \numcolors}$ trivially satisfy the booleanity constraints, edge constraints, and positivity constraint (over their respective domains), since $\pE$ satisfies these constraints. We verify these simple facts in \cref{appendix:edgepsd}. As a consequence, if $f$ has $\Pi_G^{\otimes \numcolors} f = 0$, then $\pE^{\otimes \numcolors}$ satisfies $f = 0$ as a constraint; namely, for any $g$ with $\cdeg(f g) \leq d$, it holds that $\pE^{\otimes \numcolors}[f g] = 0$.

It remains to show that $\pE'$ satisfies the sum constraints, i.e., for all $f$ with $\deg(f) \leq \frac{d}{4}$ and for every $i$, $\pE'[f^2 (\sum_c x_{i,c} - 1)] \geq 0$. Fix $i$, and let $h_i := \sum_{c} x_{i,c} \prod_{c' \ne c} (1 - x_{i,c'}) + \prod_{c} (1 - x_{i,c})$. Note that $h_i$ is the indicator of the event that $\sum_{c = 1}^\numcolors x_{i,c} \leq 1$, and when written as a polynomial, has \emph{coloring} degree $1$.

We will rely on the following two technical lemmas in our proof. The first informally shows that $\pE^{\otimes \numcolors}$ ``thinks'' that $\sum_c x_{i,c} \geq 2$ when the event indicated by $h_i$, namely ``$\sum_c x_{i,c} \leq 1$'', does not occur. Intuitively, this should clearly hold.
\begin{lemma}
\label{claim:indicatorsos}
$\pE^{\otimes \numcolors}$ satisfies the constraint $(1 - h_i) (\sum_{c} x_{i,c} - 2) \geq 0$. Namely, for every polynomial $f$ with $\cdeg(f) \leq \frac{d-2}{2}$, it holds that $\pE^{\otimes \numcolors}[f^2 (1 - h_i) (\sum_c x_{i,c} - 2)] \geq 0$.
\end{lemma}

The second lemma shows that the linear operator $\pE$ satisfies a hypercontractive inequality -- that is, the expectations of 4th powers of low-degree polynomials can be upper-bounded in terms of the expectations of their 2nd powers. Readers familiar with Fourier analysis over the hypercube may observe that the ``scaling'' in our estimate grows as $\exp (O(d \log n))$ in contrast to the $\exp (O(d))$ scaling in the usual hypercontractive inequality over the uniform measure on the Boolean hypercube. However, this worse bound will be sufficient for our purposes.

\begin{lemma}[Hypercontractivity]
\label{claim:hypercontractivity}
For any multilinear $f$ with $\cdeg(f) \leq d/4$ satisfying $f = \Pi_G^{\otimes \numcolors} f$, we have $\pE^{\otimes \numcolors}[f^4] \leq n^{O(\deg(f))} \cdot \pE^{\otimes \numcolors}[f^2]^2/(\lambda n^{-O(d)})^{2 \deg(f)}$.
\end{lemma}
We postpone the proofs of \cref{claim:indicatorsos,claim:hypercontractivity} to \cref{sec:indicatorsos,sec:hypercontractivity}, respectively, and finish the proof assuming these two claims. Let $f$ be any polynomial with $\cdeg(f) \leq \frac{d}{4}$. We lower bound $\pE^{\otimes \numcolors}[f^2 \sum_{c} x_{i,c}]$. Without loss of generality, we can assume that $f$ is multilinear, as if $f$ is not multilinear, then we can reduce it modulo the booleanity constraints. We can also assume that $f = \Pi_G^{\otimes \numcolors} f$, as if this does not hold then we write $f = f_1 + f_2$ where $f_1 = \Pi_G^{\otimes \numcolors} f_1$ and $\Pi_G^{\otimes \numcolors} f_2 = 0$, and then we observe that $\pE^{\otimes \numcolors}[(f_1 + f_2)^2 \sum_c x_{i,c}] = \pE^{\otimes \numcolors}[f_1^2 \sum_c x_{i,c}]$ (because $f_2 = 0$ is satisfied by $\pE^{\otimes \numcolors}$ as a constraint) and $\cdeg(f_1) \leq \cdeg(f) \leq \frac{d}{4}$, as the projection operation can only decrease coloring degree. We note that $\pE^{\otimes \numcolors}[h_i^2] = \pE^{\otimes \numcolors}[h_i]$, since $h_i^2 \equiv h_i$ modulo the booleanity constraints $\{x_{i,c}^2 = x_{i,c}\}$. We have that
\begin{flalign*}
&\pE^{\otimes \numcolors}[f^2 \sum_{c} x_{i,c}] = \pE^{\otimes \numcolors}[f^2 h_i \sum_c x_{i,c}] + \pE^{\otimes \numcolors}[f^2 (1 - h_i) \sum_{c} x_{i,c}] \\
&= \pE^{\otimes \numcolors}[f^2 h_i^2 \sum_c x_{i,c}^2] + \pE^{\otimes \numcolors}[f^2 (1 - h_i) \sum_{c} x_{i,c}] \ \text{(as $h_i^2 \equiv h_i$ and $x_{i,c}^2 \equiv x_{i,c}$)}  \\
& \geq  0 + \pE^{\otimes \numcolors}[f^2 (1 - h_i) \sum_{c} x_{i,c}] \ \text{(by positivity of $\pE^{\otimes \numcolors}$)} \\
&\geq \pE^{\otimes \numcolors}[f^2 \cdot 2 (1 - h_i)] \ \text{(by \cref{claim:indicatorsos})} \\
&= 2 (\pE^{\otimes \numcolors}[f^2] - \pE^{\otimes \numcolors}[f^2 h_i]) \enspace.
\end{flalign*}
Note that this is well-defined because $\pE^{\otimes \numcolors}$ is defined on each of the terms in the above inequalities since $\cdeg(f) \leq d/4 \leq (d-4)/2$ and $\cdeg(h_i) = \cdeg(\sum_c x_{i,c}) = 1$.

Now, we observe that
\begin{flalign*}
&\pE^{\otimes \numcolors}[f^2 h_i] \leq \sqrt{\pE^{\otimes \numcolors} [f^4]} \sqrt{ \pE^{\otimes \numcolors} [h_i^2]} \ \text{(by \cref{fact:cauchyschwarzcdeg})} \\
&\leq \left(\frac{n^{O(d)}}{\lambda}\right)^{2\deg(f)} \pE^{\otimes \numcolors}[f^2] \cdot \sqrt{ \pE^{\otimes \numcolors} [h_i^2]} \ \text{(by \cref{claim:hypercontractivity})}\\
&\leq \left(\frac{n^{O(d)}}{\lambda}\right)^{2\deg(f)} \pE^{\otimes \numcolors}[f^2] \cdot \sqrt{ \pE^{\otimes \numcolors} [h_i]} \ \text{(since $h_i^2 \equiv h_i$)} \enspace.
\end{flalign*}
Next, we observe that $\pE^{\otimes \numcolors}[h_i] = \pE^{\otimes \numcolors}[\prod_{c} (1 - x_{i,c})] + \sum_{c} \pE^{\otimes \numcolors}[x_{i,c} \prod_{c' \ne c}(1 - x_{i,c'})] = (1 - \pE[x_i])^\numcolors + (\numcolors - 1) (\pE[x_i] (1 - \pE[x_i])^{\numcolors-1}) \leq \numcolors \cdot e^{-\numcolors/\numlb}$ using the tensor  structure of $\pE^{\otimes \numcolors}$ and that $\pE[x_i] \geq \frac{1}{\numlb}$. Hence, 
\begin{flalign*}
&\pE^{\otimes \numcolors}[f^2 h_i] \leq \left(\frac{n^{O(d)}}{\lambda}\right)^{2\deg(f)} \pE^{\otimes \numcolors}[f^2] \cdot \numcolors \cdot e^{-\numcolors/\numlb} \\
&\implies \pE^{\otimes \numcolors}[f^2 \sum_{c} x_{i,c}] \geq 2 \pE^{\otimes \numcolors}[f^2]\left(1 - \numcolors \cdot e^{-\numcolors/\numlb}\left(\frac{n^{O(d)}}{\lambda}\right)^{2 \deg(f)}\right) \enspace.
\end{flalign*}

Now, suppose that $f$ is any polynomial with $\deg(f) \leq \frac{d}{4}$. This implies that \begin{inparaenum}[(1)] \item $\pE'[f^2 \sum_{c} x_{i,c}]$ is defined, and \item$\cdeg(f) \leq \frac{d}{4}$ \end{inparaenum}, and so we have that (using $\lambda \leq 1$)
\begin{flalign*}
\pE'[f^2 \sum_c x_{i,c}] = \pE^{\otimes \numcolors}[f^2 \sum_c x_{i,c}] \geq 2 \pE^{\otimes \numcolors}[f^2]\left(1 - \numcolors \cdot e^{-\numcolors/\numlb}\left(\frac{n^{O(d)}}{\lambda}\right)^{d/2}\right) = \pE'[f^2] \cdot 2\left(1 - \numcolors \cdot e^{-\numcolors/\numlb}\left(\frac{n^{O(d)}}{\lambda}\right)^{d/2}\right)\enspace.
\end{flalign*}
Choosing $\numcolors = O(\numlb d \log( n^d/\lambda))$, it follows that $1 - \numcolors \cdot e^{-\numcolors/\numlb}\left(\frac{n^{O(d)}}{\lambda}\right)^{d/2} \geq \frac{1}{2}$ and so $\pE'[f^2 (\sum_{c} x_{i,c} - 1)] \geq 0$ for all $f$ with $\deg(f) \leq \frac{d}{4}$. Since $\pE'$ is a degree $1 + \frac{d}{2}$ pseudo-expectation, this means that $\pE'$ satisfies the constraint $\sum_c x_{i,c} - 1 \geq 0$, which finishes the proof.

\subsection{Proof of \cref{claim:indicatorsos}}
\label{sec:indicatorsos}
Let $f$ be any polynomial with $\cdeg(f) \leq (d-2)/2$. It suffices to show that $\pE^{\otimes \numcolors}[f^2 (1 - h_i)( \sum_{c} x_{i,c} - 2)] \geq 0$. 
For $2 \leq t \leq \numcolors$, let $g_{i}^{(t)} = \prod_{c \leq t} (1 - x_{i,c})$, $g_{i,c}^{(t)} = \prod_{c' \ne c, c' \leq t} (1 - x_{i,c'})$, and $h_i^{(t)} := \sum_{c \leq t} x_{i,c} g_{i,c}^{(t)} + g_i^{(t)}$. We show by induction on $t$ that for each $t \geq 0$, it holds that $\pE^{\otimes \numcolors}[f^2 (1 - h_i^{(t)})( \sum_{c\leq t} x_{i,c} - 2)] \geq 0$ for every $f$ where the coloring degree of $f$ on the first $t$ colors is at most $(d - 2)/2$, and $\cdeg(f) \leq d/2$.

The base case is when $t = 2$. In this case, we have $1 - h_i^{(t)} = 1 - x_{i, 1}(1 - x_{i,2}) - x_{i,2}(1 - x_{i,1}) - (1 - x_{i,1})(1 - x_{i,2})  = x_{i,1} x_{i,2}$, so $\pE^{\otimes \numcolors} [f^2 (1 - h_i^{(t)})( \sum_{c \leq t} x_{i,c} - 2)] = \pE^{\otimes \numcolors}[f^2 (x_{i,1} x_{i,2}) (x_{i,1} + x_{i,2} - 2)] =\pE^{\otimes \numcolors}[f^2(2 x_{i,1} x_{i,2} - 2 x_{i,1} x_{i,2})] = 0$. Note that since $f$ has coloring degree at most $(d-2)/2$ on the first colors, $\pE^{\otimes \numcolors}$ is always defined on each of these polynomials.

We now show the induction step. We observe that $h_i^{(t+1)} = (1 - x_{i, t+1})h_{i}^{(t)} + x_{i,t+1} g_{i}^{(t)}$. Let $f$ be a polynomial that has coloring degree $\leq (d-2)/2$ on the first $t+1$ colors, and $\cdeg(f) \leq d/2$. We have
\begin{flalign*}
&\pE^{\otimes \numcolors}[f^2(1 - h_i^{(t+1)}) (x_{i, t+1} + \sum_{c \leq t+1} x_{i,c} - 2)] =\pE^{\otimes \numcolors}[f^2 \Big((1 - x_{i, t+1})(1 - h_{i}^{(t)}) + x_{i, t+1} (1 - g_{i}^{(t)})\Big) \cdot (x_{i, t+1} + \sum_{c \leq t} x_{i,c} - 2)] \\
&= \pE^{\otimes \numcolors}[f^2 (1 - x_{i, t+1})(1 - h_{i}^{(t)}) (\sum_{c \leq t} x_{i,c} - 2)] + \pE^{\otimes \numcolors}[f^2 x_{i,t + 1} (1 - g_i^{(t)}) (1 + \sum_{c \leq t} x_{i,c} - 2)] \\
&=\pE^{\otimes \numcolors}[f^2 (1 - x_{i, t+1})^2 (1 - h_{i}^{(t)}) (\sum_{c \leq t} x_{i,c} - 2)] + \pE^{\otimes \numcolors}[f^2 x_{i,t + 1}^2 (1 - g_i^{(t)}) ( \sum_{c \leq t} x_{i,c} - 1)] \enspace.
\end{flalign*}
Since $f$ has coloring degree at most $(d-2)/2$ on the first $t+1$ colors, $f \cdot (1 - x_{i,t+1})$ has coloring degree at most $(d-2)/2$ on the first $t$ colors, and also $\cdeg(f\cdot(1 - x_{i,t+1}))$ is at most $d/2$, as on the $(t+1)$-th color it has degree at most $(d - 2)/2 + 1$, and on every other color it is either at most $(d-2)/2$ or $d/2$. So, $\pE^{\otimes \numcolors}[f^2 (1 - x_{i, t+1})^2 (1 - h_{i}^{(t)}) (\sum_{c \leq t} x_{i,c} - 2)] \geq 0$ by the induction hypothesis. We also observe that $f \cdot x_{i,t+1}$ has coloring degree at most $(d-2)/2$ on the first $t$ colors, and has coloring degree at most $d/2$.

It remains to show that $\pE^{\otimes \numcolors}[f^2 (1 - g_i^{(t)}) ( \sum_{c \leq t} x_{i,c} - 1)] \geq 0$ for all $t$ and for all $f$ with $\cdeg(f) \leq d/2$ and coloring degree at most $(d-2)/2$ in the first $t$ colors. We observe that $g_{i}^{(t)} x_{i,c} \equiv 0$ for all $c \leq t$, and so it suffices to show that $\pE^{\otimes \numcolors} [f^2\sum_{c \leq t} x_{i,c} ] \geq \pE^{\otimes \numcolors} [f^2(1 - g_i^{(t)} )]$. We do this by induction on $t$. In the base case, we have $\pE^{\otimes \numcolors}[f^2 x_{i,1}] = \pE^{\otimes \numcolors}[f^2 (1 - (1 - x_{i,1}))]$. For the induction step, we have
\begin{flalign*}
&\pE^{\otimes \numcolors}[f^2 (x_{i,t+1} + \sum_{c \leq t} x_{i,c})] \geq \pE^{\otimes \numcolors}[f^2 x_{i,t+1}] + \pE^{\otimes \numcolors}[f^2 (1 - g_{i}^{(t)})] \\
&\geq \pE^{\otimes \numcolors}[f^2 x_{i,t+1} g_i^{(t)}] + \pE^{\otimes \numcolors}[f^2 (1 - g_{i}^{(t)})] \\
&= \pE^{\otimes \numcolors}[f^2 (1 - (1 - x_{i,t+1}) g_i^{(t)})] = \pE^{\otimes \numcolors}[f^2 (1 - g_{i}^{(t+1)})] \enspace,
\end{flalign*}
where we use the fact that $f x_{i,t+1} g_i^{(t)}$ has coloring degree $\leq d/2$ since $f$ has coloring degree at most $(d-2)/2$ in the first $t+1$ colors, and that $(g_i^{(t)})^2 \equiv g_i^{(t)}$ modulo the hypercube constraints. This finishes the proof.

\subsection{Proof of \cref{claim:hypercontractivity}: hypercontractivity}
\label{sec:hypercontractivity}
Let $f$ be a multilinear polynomial with $\cdeg(f) \leq d/4$ with $\Pi_G^{\otimes \numcolors} f = f$.
Suppose that:
\begin{flalign}
\pE^{\otimes \numcolors}[f^2] \geq \lambda_1 \norm{f}_2^2\enspace, \label{eq:hyper1}\\
\pE^{\otimes \numcolors}[f^4] \leq \lambda_2 \norm{f^2}_2^2\enspace, \label{eq:hyper2}\\
\norm{f^2}_2^2 \leq C \norm{f}_2^4 \enspace. \label{eq:hyper3}
\end{flalign}
Then it follows that $\pE^{\otimes \numcolors}[f^4] \leq \lambda_2 \norm{f^2}_2^2 \leq \lambda_2 C \norm{f}_2^4 \leq \frac{\lambda_2 C}{\lambda_1^2} \pE^{\otimes \numcolors}[f^2]^2$. 

\cref{eq:hyper1} follows from the following lemma with $\lambda_1 := (\lambda n^{-O(d)})^{\deg(f)}$.
\begin{lemma}[Eigenvalue lower bound for $\pE^{\otimes \numcolors}$]
\label{lem:eiglb}
Suppose that for any multilinear $g$ with $\deg(g) \leq d/2$, $\pE[g^2] \geq \lambda \norm{\Pi_G g}_2^2$. Then for any multilinear $f$ of coloring degree $\leq d/2$, it holds that $\pE^{\otimes \numcolors}[f^2] \geq (\lambda n^{-O(d)})^{\deg(f)} \cdot \norm{\Pi_G^{\otimes \numcolors} f}_2^2$.
\end{lemma}
We postpone the proof of \cref{lem:eiglb} for now, and finish the proof of \cref{claim:hypercontractivity}.

Let $g = \sum_{m} g_m \cdot m$ be a polynomial of degree $\deg(g)$ with $\cdeg(g) \leq d/2$, where $m$ is a monomial and $g_m$ is the coefficient of $g$ for the monomial $m$. Since $\pE^{\otimes \numcolors}$ satisfies the booleanity constraints, it follows that $\pE[m] \leq 1$ for all monomials $m$. Hence, $\pE[g^2] \leq \sum_{m_1, m_2} \abs{g_{m_1} g_{m_2}} = (\sum_{m} \abs{g_m})^2 \leq (n\numcolors)^{O(\deg(g))} \norm{g}_2^2$, by Cauchy-Schwarz, as $g$ is supported on at most $(n \numcolors)^{O(\deg(g))}$ distinct monomials. Since $(n\numcolors)^{O(\deg(g))} \leq n^{O(\deg(g))}$ as $\numcolors < n$, it follows that $\pE^{\otimes \numcolors}[g^2] \leq n^{O(\deg(g))} \norm{g}_2^2$, and so (setting $g = f^2$) \cref{eq:hyper2} holds with $\lambda_2 := n^{O(\deg(f))}$.

Finally, for a polynomial $g$ and monomial $m$ let $g_m$ be the coefficient of $g$ on $m$. For any $m$ of degree $\leq 2 \deg(f)$, we have that $f^2_m = \sum_{m_1, m_2: m_1 \cdot m_2 = m} f_{m_1}f_{m_2}$. We observe that this is equal to $\ip{v^{(m)}, f}$, where $v^{(m)}$ is the vector defined as $v^{(m)}_{m_2} \defeq f_{m_1}$ where $m_1 \cdot m_2 = m$ (and is $0$ if no such $m_1$ exists). It follows that $\norm{v^{(m)}}_2 \leq \norm{f}_2$, and hence that $\abs{\ip{v^{(m)}, f}} \leq \max_{v : \norm{v}_2 \leq \norm{f}_2} \abs{\ip{v,f}} = \norm{f}_2^2$. Hence, $\abs{f^2_m}^2 \leq \norm{f}_2^4$, and so $\norm{f^2}_2^2 = \sum_{m : \deg(m) \leq 2 \deg(f)} \abs{f^2_m} \leq (n \numcolors)^{O(\deg(f))} \norm{f}_2^4 = n^{O(\deg(f))} \norm{f}_4^2$, and so \cref{eq:hyper3} holds with $C := n^{O(\deg(f))}$.

Combining, we conclude that $\pE^{\otimes \numcolors}[f^4] \leq  \pE^{\otimes \numcolors}[f^2]^2/(\lambda n^{-O(d)})^{2 \deg(f)}$, which finishes the proof.

\subsubsection{Proof of \cref{lem:eiglb}: eigenvalue lower bound for $\pE^{\otimes \numcolors}$}
\label{sec:eiglb}
\parhead{Proof outline.} The proof proceeds in three steps. First, we show that the moment matrix of the independent set pseudo-expectation $\pE$, when written in a basis so that the constant polynomial $1$ is an eigenvector, has an eigenvalue lower bound of $\lambda n^{-O(d)}$. To show that this property implies the desired eigenvalue lower bound, we observe that any $f$ of total degree $\leq d$ is a linear combination of monomials that use at most $\deg(f)$ colors. Further,  (the coefficient vector of) each such monomial is a linear combination of tensor products of eigenvectors of the $\pE$ that use a ``non-$1$'' eigenvector in at most $\deg(f)$ modes of the tensor and thus is in the span of eigenvectors of $\pE^{\otimes \numcolors}$ (in the new basis) with eigenvalue at least $(\lambda n^{-O(d)})^{\deg(f)}$. This reasoning immediately implies that $f$, when written in the chosen basis, has the desired eigenvalue lower bound. To finish the proof, we argue that the change of basis does not modify $\norm{f}_2$ by too much. 

We now proceed with implementing the above proof plan. For every $S \subseteq [n]$ with $\abs{S} \leq d$, recall that we can express any multilinear polynomial $g$ with degree $\leq d$ as a linear combination of the monomials $x_S \defeq \prod_{i \in S} x_i$. Let $g_S$ be the coefficient of $g$ on the monomial $S$, so that $g = \sum_{\abs{S} \leq d} g_S x_S$. Let $e_S$ be the $S$-th standard basis vector in $\R^{{n \choose \leq d}}$. Then $g$ (as a vector of coefficients) is $\sum_{S} g_S e_S$. For $S \ne \emptyset$, define $e'_S$ as $e_S - \pE[x_S] \cdot e_{\emptyset}$. We can write $g$ uniquely in the $e'_S$ basis as $g = \sum_{S} g'_S e'_S$, where $g'_S = g_S$ for $S \ne \emptyset$, and $g'_{\emptyset} = g_{\emptyset} + \sum_{S \ne \emptyset} g_S \pE[x_S] = \pE[g]$.
Note that, if we let $x'_S := x_S - \pE[x_S]$ for $S \ne \emptyset$ and $x'_{\emptyset} := x_{\emptyset} = 1$, then $g = \sum_{S} g'_S x'_S$ as a polynomial.

Let $\M$ be the moment matrix of $\pE$ in the $x'$ basis. This matrix is indexed by sets $S, S' \subseteq [n]$ with $\abs{S}, \abs{S'} \leq d/2$, and $\M(S,S') = \pE[x'_S x'_{S'}]$, which is equal to $\pE[(x_S - \pE[x_S])(x_{S'} - \pE[x_{S'}])] = \pE[x_S x_{S'}] - \pE[x_S]\pE[x_{S'}]$ if $S,S' \ne \emptyset$, equal to $0$ if exactly one of $S,S'$ is $\emptyset$, and equal to $1$ if $S = S' = \emptyset$. This implies that $e'_{\emptyset}$ is an eigenvector of $\M$ with eigenvalue $1$. We also observe that if $g$ has degree $\leq d/2$ and $g'$ is the coefficient vector of $g$ in the $e'$ basis, then $\pE[g^2] = g'^{\top} \M g'$.

We now prove the following eigenvalue lower bound on $\M$.
\begin{claim}
$\M \succeq \lambda n^{-O(d)} \Pi_G$.
\end{claim}
\begin{proof}
Let $S$ with $\abs{S} \leq d/2$ be a set that is not an independent set in $G$, i.e.\ that $\Pi_G e_S = 0$. We observe that $\M e'_S = 0$. Indeed, the $T$-th entry of $\M e'_S$ is $\M(T, S) = \pE[x'_T x'_S] = \pE[x_T x_S] - \pE[x_T] \pE[x_S] = 0 - 0 = 0$ for $T \ne \emptyset$, and is $0$ if $T = \emptyset$ because $M(\emptyset, S) = 0$ for $S \ne \emptyset$.

Now, let $g' = \sum_{S: \abs{S} \leq d} g'_S e'_{S}$ be arbitrary. By the above, without loss of generality we may assume that $g'_S = 0$ for all $S$ that is not an independent set in $G$. Let $g$ be the corresponding polynomial in the $x$ basis, so that $g = \sum_{S} g_S x_S$, where $g_\emptyset = g'_{\emptyset}  - \sum_{S \ne \emptyset} g'_S \pE[x_S]$ and $g_S = g'_S$ for all $S \ne \emptyset$. Notice that $\pE[g] = g'_{\emptyset}$. We observe that $\Pi_G g = g$, as $g_S = g'_S = 0$ for all $S$ that is not an independent set in $G$. Now, we have that $g'^{\top} M g' = \pE[g^2] \geq \lambda \norm{\Pi_G g}_2^2 = \lambda \norm{g}_2^2$, by our eigenvalue lower bound assumption on $\pE$.

It remains to relate $\norm{g}_2^2$ and $\norm{g'}_2^2$. We have that $\norm{g'}_2^2 = \sum_{\abs{S} \leq d/2} g'^2_S = \pE[g]^2 +  \sum_{S \ne \emptyset} g'^2_S \leq \pE[g]^2 + \norm{g}_2^2 \leq \pE[g^2] + \norm{g}_2^2 \leq (n^{O(d)} + 1) \norm{g}_2^2$, as $\pE[g^2] \leq n^{O(d)} \norm{g}_2^2$ since $0 \leq \pE[x_S x_T] \leq 1$ for all $S,T$, and there are at most $n^{O(d)}$ such pairs. Hence, $g'^{\top} M g' \geq \lambda n^{-O(d)} \norm{g'}_2^2$ when $g' = \Pi_G g'$, and so $M \succeq \lambda n^{-O(d)} \Pi_G$.
\end{proof}

We have already shown that $e'_{\emptyset}$ is an eigenvector of $\M$ with eigenvalue $1$, and that the zero eigenvectors of $\M$ are the vectors $e'_{S}$ where $S$ is not an independent set in $G$. Let $f_0 = e'_{\emptyset}, f_1, \dots, f_r$ be the eigenvectors of $\M$ with nonzero eigenvalues $\lambda_0 = 1, \lambda_1, \dots, \lambda_t$, where $\lambda_i \geq \lambda n^{-O(d)}$ for $1 \leq i \leq t$. Let $\M^{\otimes \numcolors}$ be the $\numcolors$-th tensor of $\M$. Let $f_{i}^{(c)}$ denote the $i$-th eigenvector in the $c$-th component of the tensor. The eigenvectors of $\M^{\otimes \numcolors}$ are the vectors $\bigotimes_{c = 1}^\numcolors f_{i_c}^{(c)}$. We additionally observe that $\cV^{(c)} \defeq \Span{f_i^{(c)} : i > 0} = \Span{e'^{(c)}_S : \abs{S} > 0, \Pi_G e_S = e_S}$, as $f_0^{(c)} = {e'_{\emptyset}}^{(c)}$. 

Let $f$ be a multilinear polynomial with $\cdeg(f) \leq d/2$ in the variables $\{x_{i,c}\}_{i \in [n], c \in [\numcolors]}$. That is, $f$ is a vector in $\Span{ \bigotimes_{c =1}^{\numcolors} e^{(c)}_{S_c} : \abs{S_c} \leq d/2 \ \forall c \in [\numcolors]}$, where $e^{(c)}_S$ denotes the $S$-th standard basis vector in the $c$-th component of the tensor. As before, we can write $f$ as a vector $f'$ in the $e'$ basis, so $f' = \sum_{(S_1, \dots, S_\numcolors) : \abs{S_c} \leq d/2 \ \forall c \in [\numcolors]} f'_{(S_1, \dots, S_\numcolors)} \bigotimes_{c = 1}^\numcolors e'^{(c)}_{S_c}$. We again observe that $\pE^{\otimes \numcolors}[f^2] = f'^{\top} \M^{\otimes \numcolors} f'$, because the $((S_1, \dots, S_\numcolors),(T_1 ,\dots, T_\numcolors))$-th entry of $\M^{\otimes \numcolors}$ is exactly $\prod_{c = 1}^\numcolors \pE[x'_{S_c} x'_{T_c}] = \pE^{\otimes \numcolors}[\prod_{c = 1}^\numcolors x'_{S_c} x'_{T_c}]$.  Note that by the structure of the zero eigenvectors of $\M$, if $f$ satisfies $\Pi_G^{\otimes \numcolors} f = 0$, then $f$ is an eigenvector of $\M^{\otimes \numcolors}$ with eigenvalue $0$. In particular, without loss of generality we can assume that $f = \Pi_G^{\otimes \numcolors} f$, as by the above we can discard the component of $f$ in the kernel of $\Pi_G^{\otimes \numcolors}$.

Let $(S_1, \dots, S_\numcolors)$ be such that $f'_{(S_1, \dots, S_\numcolors)} \ne 0$. We must have $S_c = \emptyset$ for all but at most $\deg(f)$ of the $c$'s. This is because $f$ has degree $\deg(f)$, and so in particular every monomial in $f$ can only use at most $\deg(f)$ distinct colors. This shows that $f' \in \cV := \Span{ \bigotimes_{c \in C} \cV^{(c)} \bigotimes_{c \notin C} e'^{(c)}_{\emptyset} : C \subseteq [\numcolors], \abs{C} \leq \deg(f)}$. We observe that $\cV$ is the span of eigenvectors of $\M$ of the form $\bigotimes_{c \in C} f_{i_c}^{(c)} \bigotimes_{c \notin C} e'^{(c)}_{\emptyset}$ for $\abs{C} \leq \deg(f)$. Since each of these vectors is an eigenvector with eigenvalue at least $(\lambda n^{-O(d)})^{\abs{C}} \cdot 1^{\numcolors - \abs{C}} \geq (\lambda n^{-O(d)})^{\deg(f)}$, it follows that $f'^{\top} \M^{\otimes \numcolors} f' \geq (\lambda n^{-O(d)})^{\deg(f)}\norm{f'}_2^2$. Thus, $\pE^{\otimes k}[f^2] \geq (\lambda n^{-O(d)})^{\deg(f)}\norm{f'}_2^2$.

It remains to relate $\norm{f'}_2^2$ and $\norm{f}_2^2$. Fix $(S_1, \dots, S_\numcolors)$ with $\abs{S_c} \leq d/2$. Let $(T_1, \dots, T_\numcolors)$ with $\abs{T_c} \leq d/2$. We say that $(T_1, \dots, T_\numcolors)$ \emph{extends} $(S_1, \dots, S_\numcolors)$ if for every $c$, either $T_c = S_c$ or $T_c \ne \emptyset$ and $S_c = \emptyset$. The \emph{parity} of the extension is the parity of the number of $c$ where $T_c \ne \emptyset$ and $S_c = \emptyset$. We observe that $f_{(S_1, \dots, S_\numcolors)} = \sum_{(T_1, \dots, T_\numcolors) \ \text{extending } (S_1, \dots, S_\numcolors)} \text{(parity of extension)} \cdot f'_{(T_1, \dots, T_\numcolors)}$. This is because $e'_{(T_1, \dots, T_\numcolors)} = \bigotimes_{c = 1}^\numcolors e'_{T_c} = \bigotimes_{c = 1}^\numcolors (e_{T_c} - e_{\emptyset})$. We thus see that $\norm{f}_1 \leq \sum_{(T_1, \dots, T_\numcolors)} \abs{f'_{(T_1, \dots, T_\numcolors)}} \cdot n_{(T_1, \dots, T_\numcolors)}$, where $n_{(T_1, \dots, T_\numcolors)}$ is the number of $(S_1, \dots, S_\numcolors)$ that $(T_1, \dots, T_\numcolors)$ extends. We have shown that if $f'_{(T_1, \dots, T_\numcolors)} \ne 0$ then it must be the case that $T_{c} \ne \emptyset$ for at most $\deg(f)$ of the $c$'s. Hence, such $(T_1, \dots, T_\numcolors)$ can only extend at most $2^{\deg(f)}$ of the $(S_1, \dots, S_\numcolors)$'s, as each of the $(S_1, \dots, S_\numcolors)$'s is obtained by changing a subset of the $T_c$'s to be empty. Hence, $\norm{f}_1 \leq 2^{\deg(f)} \norm{f'}_1$. Since $f'$ has at most $(n\numcolors)^{\deg(f)} \leq n^{2 \deg(f)}$ nonzero coefficients, we get that $\norm{f}_2 \leq n^{ \deg(f)} \cdot 2^{\deg(f)} \norm{f'}_2$, and so we conclude that $\norm{f}_2^2 \leq n^{O(\deg(f))} \norm{f'}_2^2$, which finishes the proof.

\section{Proof of \cref{thm:coloringlb}: coloring lower bound}
\label{sec:fmain}
We now prove \cref{thm:coloringlb} (restated below in the language of pseudo-expectations) from \cref{thm:reduction}. In this section, we assume familiarity with the planted clique pseudo-expectation of \cite{DBLP:conf/focs/BarakHKKMP16}. 
\begin{theorem*}[\cref{thm:coloringlb}, restated]
For sufficiently large $n$, for any $\eps \in (\Omega(\sqrt{\frac{1}{\log n}}), \frac{1}{2})$, with probability $1 - 1/\poly(n)$ over the draw of $G \sim G(n,1/2)$, there is a degree $d = O(\eps^2 \log n)$ coloring pseudo-expectation $\pE$ using $k = n^{\frac{1}{2} + \eps}$ colors.
\end{theorem*}

We begin by recalling the main theorem of~\cite{DBLP:conf/focs/BarakHKKMP16}. 
\begin{theorem}[\cite{DBLP:conf/focs/BarakHKKMP16}]
There is an absolute constant $C$ such that for $n$ sufficiently large, $C /\sqrt{\log n} \leq \eps < \frac{1}{2}$, $\omega = n^{\frac{1}{2} - \eps}$, and $d = (\eps/C)^2 \log n$,  with probability $1 - 1/\poly(n)$ over $G \sim G(n,1/2)$, the operator $\pE_{G}$ defined in \cite{DBLP:conf/focs/BarakHKKMP16} satisfies:
\begin{enumerate}
\item $\pE_G[1] = 1 \pm n^{-\Omega(\eps)}$,
\item $\pE_G[ \sum_{i} x_i] = \omega (1 \pm n^{-\Omega(\eps)})$,
\item $\pE_G[x_S] = 0$ for all $\abs{S} \leq d$ that is not a clique in $G$,
\item $\pE_G[f^2] \geq \lambda \norm{\Pi'_G f}_2^2$ where $\lambda = \Omega\left(\left(\frac{\omega}{n} \right)^{d + 1}\right)$ and $\Pi'_G$ is the projection onto $x_S$ for $S$ a clique in $G$.
\end{enumerate}
\end{theorem}
We first observe that if $G \sim G(n,1/2)$, then the complement graph $\bar{G} \sim G(n,1/2)$ also, and moreover $\pE_G$ will satisfy the independent set constraints as $\pE_G[x_S] = 0$ for $S$ that is not a clique in $G$, which is equivalent to $S$ not being an independent set in $\bar{G}$. We also note that $\Pi'_G = \Pi_{\bar{G}}$, and that the final pseudo-expectation is obtained by setting $\pE[x_S] := \pE_G[x_S]/\pE_G[1]$; this is done so that the normalization condition $\pE[1] = 1$ is satisfied. 

We thus see that $\pE$ satisfies the second additional condition of \cref{thm:reduction}. Hence, in order to apply \cref{thm:reduction} to conclude \cref{thm:coloringlb}, it suffices to argue that with high probability over $G$, it holds that $\pE_G[x_i] \geq \frac{\omega}{n}(1 - n^{-\Omega(\eps)})$ for all $i$. Indeed, if this holds then we have $\pE[x_i] \geq \frac{\omega}{n} (1 - n^{-\Omega(\eps)})$ also, and then we can apply \cref{thm:reduction} with $k = \frac{n}{\omega} (1 + n^{-\Omega(\eps)})$ which finishes the proof. Thus, it remains to prove the following claim.

\begin{claim}
\label{claim:xi}
For each $i$, $\pE_G[x_i] \geq \frac{\omega}{n} \cdot (1 - n^{-\Omega(\eps)})$ with probability $1 - n^{-\log n}$.
\end{claim}
\begin{proof}
 We have that $\pE_G[x_i] := \sum_{T \subseteq {[n] \choose 2} : \abs{V(T)} \leq \tau} \left(\frac{\omega}{n}\right)^{\abs{V(T) \cup \{i\}}} \chi_T(G)$, where $\tau \leq (\eps/C) \log n$. The $T = \emptyset$ term always contributes $\frac{\omega}{n}$. The other terms all have $\abs{V(T)} \geq 2$. Let $H_1$ be the set of $T$ such that $i \in V(T)$, and let $H_2$ be the set of $T$ such that $i \notin V(T)$. Let $H_{1}^{(t)}$ be the set of $T \in H_1$ with $\abs{V(T)} = t$, and similarly for $H_2^{(t)}$. Each set $H_{1}^{(t)}$ can be partitioned into families $\{\cT_{1,r}^{(t)}\}_{r = 1}^{p_{1,t}}$ where $T$ and $T'$ are in the same family if there is a permutation $\sigma \colon [n] \to [n]$ such that $T = \sigma(T')$, or equivalently if $T$ and $T'$ are isomorphic. Similarly, each set $H_{2}^{(t)}$ can be partitioned into families $\{\cT_{2,r}^{(t)}\}_{r = 1}^{p_{2,t}}$.

We thus have
\begin{flalign*}
\bigabs{\pE_G[x_i] - \frac{\omega}{n}} \leq \sum_{t = 2}^{\tau} \left[\left(\frac{\omega}{n}\right)^{t} \sum_{r = 1}^{p_{1,t}}\bigabs{\sum_{T \in \cT_{1,r}^{(t)}} \chi_T(G)} + \left(\frac{\omega}{n}\right)^{t+1} \sum_{r = 1}^{p_{2,t}}\bigabs{\sum_{T \in \cT_{2,r}^{(t)}} \chi_T(G)}\right] \enspace.
\end{flalign*}

\begin{lemma}
\label{lem:chiconc}
Let $\cT$ be a family of subsets of ${[n] \choose 2}$ such that $\abs{V(T)} = t$ for every $T \in \cT$, and for every $T, T' \in \cT$, there exists $\sigma \colon [n] \to [n]$ such that $T = \sigma(T')$. Let $S = \cap_{T \in \cT} V(T)$. Then for every $s \geq 0$ and even $\ell$, 
\begin{equation*}
\Pr_{G \sim G(n,1/2)}\left[ \bigabs{\sum_{T \in \cT} \chi_T(G)} \leq s\right] \geq 1 - \frac{n^{(t - \abs{S}) \ell/2} \cdot (t\ell)^{t \ell}}{s^{\ell}} \enspace.
\end{equation*}
\end{lemma}
We postpone the proof of \cref{lem:chiconc} to \cref{sec:chiconc}, and now use it to finish the proof of \cref{claim:xi}. Applying \cref{lem:chiconc} with $\ell = (\log n)^2$, we get
\begin{flalign*}
&\bigabs{\sum_{T \in \cT_{1,r}^{(t)}} \chi_T(G)} \leq n^{(t-1)/2} (\log n)^{3t} \text{ with probability } \geq 1 - 2^{-t \log^2 n(\log \log n - \log t)}\\
&\bigabs{\sum_{T \in \cT_{2,r}^{(t)}} \chi_T(G)} \leq n^{t/2} (\log n)^{3t} \text{ with probability } \geq 1 - 2^{-t \log^2 n(\log \log n - \log t)} \enspace.
\end{flalign*}
We observe that $p_{1,t}$ and $p_{2,t}$ are both at most $2^{t^2}$, as an equivalence class with $t$ vertices is uniquely determined by a graph on $t$ vertices. By union bound, we see that the above holds for all equivalence classes $\cT_{1,r}^{(t)}$ and $\cT_{2,r}^{(t)}$ with probability at least $1 - 2 \sum_{t = 2}^{\tau} 2^{t^2 - t \log^2 n(\log \log n - \log t)}$. Since $t \leq \tau \leq (\eps/C) \log n$, it follows that 
\begin{flalign*}
 &\sum_{t = 2}^{\tau} 2^{t^2 - t \log^2 n(\log \log n - \log t)} = \sum_{t = 2}^{\tau} 2^{t(t - \log^2n (\log \log n - \log t))} \leq \sum_{t = 2}^{\tau} 2^{t( \frac{\eps}{C} \log n - (\log^2 n) (\log \log n - \log \frac{\eps}{C} - \log \log n))}\\
 &\leq \tau \cdot 2^{2 \log n \cdot (\frac{\eps}{C} - \log \frac{C}{\eps} \cdot \log n)} \leq n^{- \log n} \enspace,
\end{flalign*}
as $\frac{C}{\eps} \geq C \geq 16$. Thus, with probability at least $1 - n^{-\log n}$, we have
\begin{flalign*}
&\bigabs{\pE_G[x_i] - \frac{\omega}{n}} \leq \sum_{t = 2}^{\tau} \left[\left(\frac{\omega}{n}\right)^{t} 2^{t^2}  n^{(t-1)/2} (\log n)^{3t} + \left(\frac{\omega}{n}\right)^{t+1} 2^{t^2}n^{t/2} (\log n)^{3t}  \right]\\
&\leq \frac{2}{\sqrt{n}} \sum_{t = 2}^{\tau} \left(\frac{\omega}{n}\right)^{t} 2^{t^2}(\log n)^{3t} n^{t/2} = 2 \left(\frac{\omega}{n}\right)  \sum_{t = 2}^{\tau} n^{(t - 1)/2} n^{-(t-1)\eps} 2^{t^2}(\log n)^{3t} n^{t/2} n^{-1/2} \\
&\leq 2 \left(\frac{\omega}{n}\right)  \sum_{t = 2}^{\tau}n^{-(t-1)\eps} 2^{t^2}(\log n)^{3t} \leq \left(\frac{\omega}{n}\right) \cdot \max_{2 \leq t \leq \tau} 2 n^{-(t-1)\eps} 2^{t^2}(\log n)^{3t + 1}\\
& \leq \left(\frac{\omega}{n}\right) \cdot 2 n^{\eps} (\log n)\max_{2 \leq t \leq \tau} (n^{-\eps} 2^{\tau}(\log n)^{3})^t \leq \left(\frac{\omega}{n}\right) \cdot n^{\eps} n^{2\eps/K} \max_{2 \leq t \leq \tau} (n^{-\eps} n^{\eps/C} \cdot n^{3 \eps/K})^t \\
&\leq  \left(\frac{\omega}{n}\right) \cdot n^{\eps} n^{2\eps/K}  n^{-2\eps(1 - 1/C - 3/K)} =  \left(\frac{\omega}{n}\right) \cdot  n^{-\eps (1 - 2/C - 8/K)} \leq  \left(\frac{\omega}{n}\right) \cdot n^{- \eps/2}\enspace,
\end{flalign*}
as $\eps \geq C/\sqrt{\log n} \geq K \log \log n/\log n$ for $K \geq 32$ and $\tau \leq (\eps/C) \log n$. Hence, with probability $1 - 1/n^{\log n}$, we have that $\pE_G[x_i] = \frac{\omega}{n} (1 \pm n^{-\eps/2})$, which completes the proof.
\end{proof}
\subsection{Proof of \cref{lem:chiconc}}
\label{sec:chiconc}
Let $\ell \in \N$ be even. We have that
\begin{flalign*}
\E_{G \sim G(n,1/2)} \bigabs{\sum_{T \in \cT} \chi_T(G)}^{\ell} = \E_{G \sim G(n,1/2)} \big(\sum_{T \in \cT} \chi_T(G)\big)^{\ell} = \sum_{T_1, \dots, T_{\ell} \in \cT} \E_{G \sim G(n,1/2)}\prod_{i = 1}^{\ell} \chi_{T_i}(G) \enspace.
\end{flalign*}
We have that $\E_{G \sim G(n,1/2)}\prod_{i = 1}^{\ell} \chi_{T_i}(G)  = 1$ iff $\bigoplus_{i = 1}^{\ell} T_i = \emptyset$, that is, every edge in the multiset $\cup_{i = 1}^{\ell} T_i$ appears an even number of times, and otherwise the term is $0$. Since every edge in the multiset appears an even number of times, every vertex also appears an even number of times in $\cup_{i = 1}^{\ell} V(T_i)$, and hence every vertex appears at least twice. Since $S \subseteq V(T_i)$ for all $i$, every vertex in $S$ appears exactly $\ell$ times. So, the number of distinct vertices in $\cup_{i = 1}^{\ell} (V(T_i) \setminus S)$ is at most $(t - \abs{S})\cdot \ell/2$. Each tuple $(T_1, \dots, T_{\ell})$ with this property can thus be chosen by \begin{inparaenum}[(1)] \item selecting $(t - \abs{S})\cdot \ell/2$ distinct vertices $S'$ (at most $n^{(t - \abs{S}) \ell/2}$ choices),  and then \item choosing injections $\sigma_i \colon V(T) \to S'$ and setting $T_i = \sigma_i(T)$, where $T \in \cT$ is an arbitrary fixed element (at most $(\abs{S'}^{t})^{\ell} \leq (t \ell)^{t \ell}$ choices)\end{inparaenum}.
Thus, we get $\E_{G \sim G(n,1/2)} \bigabs{\sum_{T \in \cT} \chi_T(G)}^{\ell}  \leq n^{(t - \abs{S}) \ell/2}(t \ell)^{t \ell}$. By Markov's inequality, it follows that $\Pr_{G \sim G(n,1/2)}\left[\bigabs{\sum_{T \in \cT} \chi_T(G)}  > s\right]  = \Pr_{G \sim G(n,1/2)}\left[\bigabs{\sum_{T \in \cT} \chi_T(G)}^{\ell}  > s^{\ell}\right] \leq \frac{n^{(t - \abs{S}) \ell/2}(t \ell)^{t \ell}}{s^{\ell}}$, which completes the proof.

\section*{Acknowledgements}

We thank Xinyu Wu for taking part in early stages of this research, and the anonymous reviewers for providing valuable feedback.

This research was supported in part by:
the NSF CAREER Award (\#2047933),
the NSF Graduate Research Fellowship Program (under Grant No.\ DGE1745016),
and
the ARCS Foundation.
Any opinions, findings, and conclusions or recommendations expressed in this material are those of the author(s) and do not necessarily reflect the views of the National Science Foundation.

\phantomsection

\bibliographystyle{amsalpha}
\bibliography{bib/custom,bib/custom2,bib/mathreview,bib/dblp,bib/scholar}  
  \appendix
  
\section{Satisfying the booleanity, edge and positivity constraints}
\label{appendix:edgepsd}
We prove the following three simple claims.

\begin{claim}
\label{claim:bool}
$\pE^{\otimes \numcolors}$ and $\pE'$ satisfy the booleanity constraints $\{x_{i,c}^2 =  x_{i,c} : i \in [n], c \in [\numcolors]\}$.
\end{claim}
\begin{claim}
\label{claim:edge}
$\pE^{\otimes \numcolors}$ and $\pE'$ satisfy the edge constraints $\{x_{i,c} x_{j,c} = 0 : (i,j) \in E(G), c \in [\numcolors]\}$.
\end{claim}

\begin{claim}
\label{claim:psd}
$\pE^{\otimes \numcolors}$ and $\pE'$ satisfy the positivity constraint.
\end{claim}

\begin{proof}[Proof of \cref{claim:bool}]
Since $\pE'$ is obtained by restricting $\pE^{\otimes \numcolors}$ to a smaller domain, it suffices to show that $\pE^{\otimes \numcolors}$ satisfies the constraints. We observe that $\pE^{\otimes \numcolors}$ satisfies the above constraints if and only if for all monomials $\prod_{c = 1}^\numcolors \prod_{i \in S_c} x_{i,c}^{\alpha_{i,c}}$ (where each $\alpha_{i,c} \geq 1$), it holds that $\pE^{\otimes \numcolors}[\prod_{c = 1}^\numcolors \prod_{i \in S_c} x_{i,c}^{\alpha_{i,c}}] = \pE^{\otimes \numcolors}[\prod_{c = 1}^\numcolors \prod_{i \in S_c} x_{i,c}]$. We have that $\pE^{\otimes \numcolors}[\prod_{c = 1}^\numcolors \prod_{i \in S_c} x_{i,c}^{\alpha_{i,c}}] = \prod_{c = 1}^\numcolors \pE[\prod_{i \in S_c} x_{i}^{\alpha_{i,c}}] = \prod_{c = 1}^\numcolors \pE[\prod_{i \in S_c} x_{i}] = \pE^{\otimes \numcolors}[\prod_{c = 1}^\numcolors \prod_{i \in S_c} x_{i,c}]$, as $\pE$ satisfies the constraints $x_i^2 = x_i$, and so we are done.
\end{proof}

\begin{proof}[Proof of \cref{claim:edge}]
Since $\pE'$ is obtained by restricting $\pE^{\otimes \numcolors}$ to a smaller domain, it suffices to show that $\pE^{\otimes \numcolors}$ satisfies the constraints. We observe that $\pE^{\otimes \numcolors}$ satisfies the above constraints if and only if for all multilinear monomials $\prod_{c = 1}^\numcolors x_{S_c,c}$ of coloring degree at most $d-2$, it holds that $\pE^{\otimes \numcolors} [x_{i,c} x_{j,c} \prod_{c' = 1}^\numcolors x_{S_c', c'}] = 0$. This is because by \cref{claim:bool}, we can reduce any polynomial modulo the booleanity constraints to make it multilinear. Using the tensor product structure, we have $\pE^{\otimes \numcolors} [x_{i,c} x_{j,c} \prod_{c' = 1}^\numcolors x_{S_c', c'}] = \prod_{c' \ne c} \pE[x_{S_{c'}}] \cdot \pE[x_{S_c} x_i x_j] = \prod_{c' \ne c} \pE[x_{S_{c'}}] \cdot 0 = 0$, since $\pE$ satisfies the edge constraints. This completes the proof.
\end{proof}

\begin{proof}[Proof of \cref{claim:psd}]
Since $\pE'$ is obtained by restricting $\pE^{\otimes \numcolors}$ to a smaller domain, it suffices to prove the claim only for $\pE^{\otimes \numcolors}$. Let $\M$ be the moment matrix of $\pE$. That is, $\M$ is the matrix indexed by sets $(S,T)$ with $\abs{S}, \abs{T} \leq d/2$ and $\M(S,T) := \pE[x_S x_T]$. We note that for any $f \in \cP^n_{d/2}$, $\pE[f^2] = f^{\top} M f$, where we interpret $f$ as a vector of coefficients in the second expression. The moment matrix of $\pE^{\otimes \numcolors}$ is indexed by tuples of sets $((S_1, \dots, S_\numcolors), (T_1, \dots, T_\numcolors))$ where $\abs{S_c}, \abs{T_c} \leq d/2$ for all $c \in [\numcolors]$. We observe that the moment matrix of $\pE^{\otimes \numcolors}$ is $\M^{\otimes \numcolors}$, as the $((S_1, \dots, S_\numcolors), (T_1, \dots, T_\numcolors))$-th entry is $\pE^{\otimes \numcolors}[\prod_{c = 1}^\numcolors x_{S_c, c} x_{T_c, c}] = \prod_{c = 1}^\numcolors \pE[x_{S_c} x_{T_c}] = \prod_{c = 1}^\numcolors \M(S_c, T_c)$. We also note that for any $f$ with $\cdeg(f) \leq d/2$, it holds that $\pE^{\otimes \numcolors}[f^2]  = f^{\top} \M^{\otimes \numcolors}f \geq 0$, as the tensor product of a positive semidefinite matrix is also positive semidefinite. This shows that $\pE^{\otimes \numcolors}[f^2] \geq 0$ for all $f$ with $\cdeg(f) \leq d/2$, which finishes the proof.
\end{proof}

\section{Tightness of degree in \cref{thm:coloringlb}}
\label{sec:independentsetrefutation}
In this section, we prove the following lemma, showing that the upper bound on $d$ in \cref{thm:coloringlb} is tight up to constant factors.
\begin{lemma}
With high probability over $G \sim G(n,1/2)$, there is no degree $8(1 + o(1)) \log_2 n$ coloring pseudo-expectation for $G$ using $\numcolors \leq \frac{n}{e \cdot 2(1 + o(1))\log_2 n}$ colors.
\end{lemma}
Let $t = 2(1 + o(1)) \log_2 n$. We show that with high probability over $G \sim G(n,1/2)$, there is no degree $4t$ coloring pseudo-expectation for $G$ using $\numcolors \leq \frac{n}{e t}$ colors. We first observe that with high probability, the maximum independent set in $G$ has size at most $t$. Suppose that we draw $G \sim G(n,1/2)$ such that this holds, and suppose that such a pseudo-expectation $\pE'$ exists. We observe that there is a natural action of permutations $\sigma \colon [\numcolors] \to [\numcolors]$ on $\pE'$, given by $\pE'^{(\sigma)}[\prod_{c = 1}^\numcolors \prod_{i \in S_c} x_{i,c}] := \pE'[\prod_{c = 1}^\numcolors \prod_{i \in S_c} x_{i, \sigma(c)}]$. Let $\pE'' := \E_{\sigma} \pE'^{(\sigma)}$ be the pseudo-expectation obtained by averaging over all $\sigma$. We then have that $\pE''$ satisfies the coloring constraints and is symmetric with respect to the color classes, e.g.\ that $\pE''[x_{i,c}] = \pE''[x_{i,c'}]$ for all $c,c' \in [\numcolors]$. This implies that $\pE''[x_{i,1}] = \frac{1}{\numcolors} \sum_{c = 1}^\numcolors \pE''[x_{i,c}] \geq \frac{1}{\numcolors} \cdot 1$. Let $\pE$ be the projection of $\pE''$ onto the first color, so that $\pE[\prod_{i \in S} x_i] := \pE''[\prod_{i \in S} x_{i,1}]$. We then see that $\pE$ is a degree $4t$ independent set pseudo-expectation with $\pE[\sum_i x_i] \geq \omega$, where $\omega := \frac{n}{\numcolors} \geq e t$.

To complete the proof, we show the following lemma.
\begin{lemma}
Suppose that the maximum independent set in $G$ has size $\leq t$. Then there is no degree $4t$ independent set pseudo-expectation $\pE$ for $G$ with $\pE[\sum_i x_i] = \omega \geq e t$.
\end{lemma}
\begin{proof}
Suppose that such a pseudo-expectation $\pE$ exists. Let $f = \sum_i x_i$, and let $\ell \in \N$ be the smallest integer so that $2^{\ell} \geq 2t$. Note that $2^{\ell} \leq 4t$ must hold also. By Cauchy-Schwarz, we have
\begin{flalign*}
&\pE[f^{2^{\ell}}] \geq (\pE[f^{2^{\ell-1}}])^2 \enspace, \\
&\pE[f^{2^{\ell - 1}}] \geq \pE[f^{2^{\ell - 2}}]^2 \geq \dots \geq \pE[f]^{2^{\ell-1}} \enspace, \\
&\implies \pE[f^{2^{\ell}}] \geq (\pE[f^{2^{\ell-1}}]) \cdot (\pE[f])^{2^{\ell - 1}} = \pE[f^{2^{\ell-1}}]  \cdot \omega^{2^{\ell - 1}} \enspace.
\end{flalign*}
Note that each polynomial above has degree at most $2^{\ell} \leq 4t$, so the above pseudo-expectations are all well-defined. Now, we observe that 
\begin{flalign*}
&\pE[f^{2^{\ell - 1}}] = \pE[\sum_{S \subseteq [n] : \abs{S} \leq 2^{\ell - 1}} c_S x_S] =\sum_{S : \abs{S} \leq t, \text{ $S$ indep set in $G$}} c_S \pE[x_S] \enspace, \\
&\pE[f^{2^\ell}] = \sum_{S : \abs{S} \leq t, \text{ $S$ indep set in $G$}} c'_S \pE[x_S]
\end{flalign*}
where the coefficients $c_S$ and $c'_S$ are each nonnegative integers. Notice that $c'_S \leq \abs{S}^{2^{\ell}}$, as every contribution to $x_S$ is made by choosing an $i \in S$ from each of the $\sum_i x_i$ factors. We also observe that $c_S \geq \abs{S}^{2^{\ell - 1} - \abs{S}} \cdot (\abs{S}!)$, as we can choose each $i \in S$ exactly once from the first $\abs{S}$ factors, and then select an arbitrary $i \in S$ from the remaining $2^{\ell - 1} - \abs{S}$ factors. Note that here we use the fact that $\abs{S} \leq t \leq 2^{\ell - 1}$ always holds.
Fix $S$, and let $s = \abs{S}$. We observe that
\begin{flalign*}
\frac{c'_S}{c_S} \leq \frac{s^{2^{\ell}}}{s^{2^{\ell - 1} - s} \cdot s!} \leq s^{2^{\ell - 1}} \cdot s^s \cdot \frac{1}{\sqrt{2 \pi} \cdot s^{s + \frac{1}{2}} e^{-s}} < s^{2^{\ell - 1}} \cdot e^s \leq (e \cdot s)^{2^{\ell - 1}} \leq \omega^{2^{\ell - 1}} \enspace,
\end{flalign*}
using Stirling's approximation and the fact that $\omega \geq e t \geq e s$. It therefore follows that $(c'_S - c_S \omega^{2^{\ell - 1}}) \pE[x_S] < 0$. Hence,
\begin{flalign*}
\pE[f^{2^\ell}] - \pE[f^{2^{\ell - 1}}] \cdot \omega^{2^{\ell - 1}} = \sum_{S : \abs{S} \leq t, \text{ $S$ indep set in $G$}} (c'_S - c_S \omega^{2^{\ell-1}}) \pE[x_S] < 0 \enspace,
\end{flalign*}
which is a contradiction. 
\end{proof}

\end{document}